\documentclass[12pt]{article}

\usepackage[utf8]{inputenc}
\usepackage{amsmath,amssymb,amsfonts,amsthm}
\usepackage{mathtools}
\usepackage{graphicx}
\usepackage{epstopdf}
\usepackage{hyperref}
\usepackage{float}
\usepackage{subcaption}
\usepackage{multirow}
\usepackage{enumerate}
\usepackage{csquotes}
\usepackage{lipsum}
\usepackage{caption}
\usepackage{lineno}
\usepackage{anysize}
\usepackage{color,soul}

\usepackage{xr-hyper}   
\usepackage{hyperref}
\usepackage{authblk}

\newtheorem{theorem}{Theorem}[section]
\newtheorem{lemma}[theorem]{Lemma}

\title{Effect of wind on prey-predator dynamics with group defense and additional food strategy}

\author[1]{Anushree Hazra}

\author[2]{Aniket Banerjee}

\author[1]{Debaldev Jana\thanks{Corresponding author: \texttt{debaldevjana.jana@gmail.com}}}

\affil[1]{Department of Mathematics, 
Ramakrishna Mission Vivekananda Educational and Research Institute, 
Belur Math, West Bengal 711202, India}

\affil[2]{Sorbonne Université, 
CNRS, Laboratoire Jacques-Louis Lions (LJLL), 
75005 Paris, France}

\begin{document}
\date{} 
\maketitle

\begin{abstract}
Wind plays a crucial role in changing prey defense strategy and predation efficiency. We develop and analyze a prey-predator model that incorporates wind-driven environmental effects, prey group defense, and an additional food strategy for the predator. Wind intensity is assumed to modulate predation efficiency, while prey aggregation reduces predation pressure at high densities, and supplementary food alters predator growth independent of prey abundance. We investigate the existence and stability of biologically feasible equilibrium points. The analysis reveals that wind strength and additional food supply can induce stability switches, oscillatory dynamics via Hopf bifurcation, and more complex behaviors including saddle–node and Bogdanov–Takens bifurcations. Our results demonstrate that environmental forcing and resource supplementation jointly shape predator persistence and population fluctuations, providing theoretical insights into ecological management strategies involving food enrichment under variable environmental conditions.
    
\end{abstract}

\section{Introduction}
Prey–predator interactions are fundamental components of ecological systems and play a key role in shaping population dynamics and ecosystem structure \cite{wooster2026predator,lima1998nonlethal,beauchamp2007predator}.
In natural ecosystems, however, population dynamics are strongly influenced by environmental factors \cite{tablado2014environmental,domenici2007environmental,sih1994prey,marrow1996evolutionary} that are often neglected in basic models. Among these factors, wind plays a significant role in shaping ecological interactions, particularly in open terrestrial and aquatic environments \cite{barman2022impact}. Wind can affect prey–predator dynamics by altering movement patterns, dispersal rates, encounter frequencies, and foraging efficiency. For instance, wind-driven transport may facilitate or hinder predator–prey encounters, modify spatial distribution, and influence the energetic cost of movement for both species.

Motivated by these ecological considerations, several studies have emphasized the importance of incorporating environmental forcing into population models \cite{cuddington2002predator,engstrom1970predation,barman2023modelling,thirthar2026dynamics}. In this context, wind can be modeled as an external factor that modifies interaction terms, movement behavior, or species growth rates. Such modifications lead to nonlinear dynamical systems whose qualitative behavior may differ significantly from that of classical prey–predator models.

The role of additional food for predators has received considerable attention, particularly in biological control theory \cite{srinivasu2007biological,srinivasu2011role,srinivasu2018biological}. The provision of supplementary food, either naturally occurring or artificially supplied, modifies both the predator's numerical response and the functional response \cite{prasad2013dynamics,gurubilli2017global}. Several studies have demonstrated that the quality and quantity of additional food can either stabilize or destabilize predator-prey interactions \cite{srinivasu2010time, srinivasu2011role}. It has also been shown that supplementary food may lead to counterintuitive outcomes such as prey extinction or the emergence of complex bifurcation structures depending on parameter regimes \cite{parshad2023additional,VERMA2026100114}. However, most of these studies assume constant environmental conditions and do not explicitly incorporate adaptive prey defense mechanisms within the same modeling framework.

Although substantial progress has been made in modeling additional food, group defense, and nonlinear functional responses independently, an integrated framework combining these mechanisms with wind-modulated predation remains largely unexplored. In particular, the interaction between wind-dependent olfactory efficiency and supplementary food may generate novel stability thresholds and higher codimension bifurcations that cannot be captured by models lacking environmental modulation.

We develop a prey-predator model to understand the dynamical change in the species population when the prey adopts a group defense strategy as an anti-predator response in the presence of an additional food source for the predator. The effect of wind is studied and is important to understand its implications for population dynamics.

We assume the prey population follows logistic growth in the absence of the predator. The ecosystem has a carrying capacity $k$, which takes into account that there is a limit to the amount of resources for the populations in the ecosystem, which stops the unbounded growth of the prey population in the absence of a predator. The intrinsic growth rate of the prey population is assumed to be the positive constant $r$.
The logistic growth of the prey population simplifies and takes into account the growth, natural mortality, and intraspecific competition of the prey population. 

\begin{equation}\label{eq:1}
    \begin{cases}
        &\dfrac{dx}{dt}=rx\left ( 1-\frac{x}{k}\right) - yp(x), \\
        &\dfrac{dy}{dt}= q(x)y-dy,\\
    &x(0)>0,\hspace{2mm}y(0)>0.
    \end{cases}
\end{equation}

We assume the group defense strategy to be modeled as Holling type IV or Monod-Haldane type functional response(\cite{FREEDMAN1986493,sasmal2020dynamics}). The functional response to observe the study of group defense can be stated as $p(x)=\dfrac{\beta x}{a+bx+x^2}$ where $\beta$ is the rate of predation of the prey, $b$ is the rate of group defense by the prey population. As $b$ increases, $p(x)$ decreases, thereby reducing predation. So, group defense as an anti-predator response sensitivity is studied through the parameter $b$.

The additional food for the predator is assumed to exist or maybe be provided externally with the quantity $\xi$ and quality $\frac{1}{\alpha}$. The additional food source decreases the predation function response, i.e., $p(x)$ decreases with an increase in $\xi$, as the predator has an alternate food source for predation, while it increases the predator growth by providing nutrients, i.e., $q(x) $ increases with the increase of $\xi$. Such models have been studied extensively in the literature \cite{srinivasu2010time,srinivasu2011role,parshad2023additional}.

Thus, a developed model with the effect of group defense and the presence of additional food can be given by,

\begin{equation}\label{eq:2}
    \begin{cases}
        &\dfrac{dx}{dt}= rx\left(1-\frac{x}{k}\right) - \dfrac{ x y}{1+\alpha \xi + bx + x^2},\\
        &\dfrac{dy}{dt}= \dfrac{\beta (x+\xi)y}{1+\alpha \xi + bx + x^2} - dy,\\
         &x(0)>0,\hspace{2mm}y(0)>0.
    \end{cases}
\end{equation}

Beyond biotic interactions, abiotic environmental factors play a crucial role in shaping predator-prey encounters. Wind, in particular, influences olfactory communication \cite{banks2016predator,hilton1999choice}, dispersal patterns \cite{damschen2014fragmentation,nathan2011mechanistic}, and encounter rates in terrestrial and aerial ecosystems \cite{thomas2003aerial}. Empirical studies suggest that wind speed affects the transmission of chemical cues in a non-monotonic manner. At very low wind speeds, odor plumes disperse poorly, limiting detection, while at very high speeds, turbulence disrupts odor gradients and reduces signal reliability. Consequently, olfactory efficiency often exhibits an optimal intermediate wind speed. Despite such biological evidence, wind effects are rarely mechanistically incorporated into deterministic prey-predator models. Environmental variability is often modeled as stochastic noise or periodic forcing, without an explicit functional dependence between wind speed and predation efficiency.

Motivated by these considerations, we propose and analyze a prey-predator model that incorporates logistic prey growth, group defense represented via a Monod-Haldane type functional response, additional food of specified quality and quantity for the predator, and a wind-dependent modulation of predation efficiency described through a biologically motivated function. The wind function reflects the empirically observed relationship between wind speed and olfactory signal reception and directly modifies both predation and predator growth terms. In addition, a density-dependent mortality term associated with supplementary food consumption is included to capture potential ecological costs.

Thus, after modification, our model for the study is given as follows:

\begin{equation}\label{eq:5}
    \begin{cases}
        &\dfrac{dx}{dt}= rx\left(1-\frac{x}{k}\right) - \dfrac{\phi(w) x y}{1+\alpha \xi + bx + x^2},\\
        &\dfrac{dy}{dt}= \dfrac{\beta \phi(w) (x+\xi)y}{1+\alpha \xi + bx + x^2} - dy -c\xi \phi(w) y^2,\\
         &x(0)>0,\hspace{2mm}y(0)>0.
    \end{cases}
\end{equation}

The main objective of this work is to provide a rigorous analytical and numerical investigation of the resulting dynamical system. We establish positivity and uniform boundedness of solutions to ensure ecological feasibility. We determine the existence and local stability conditions of all equilibria, including trivial, boundary, and interior steady states. We further investigate local bifurcations, including transcritical, saddle-node, and Hopf bifurcations, with respect to biologically relevant parameters. Moreover, we analyze higher-codimension global bifurcations, including cusp singularities and Bogdanov--Takens bifurcations, thereby revealing the possible emergence of unstable limit cycles and homoclinic loops. Numerical simulations are performed to illustrate how different wind-dependent functional forms influence extinction, coexistence, and oscillatory dynamics.

The paper is organized as follows. Section~\ref{S3} establishes positivity and uniform boundedness of solutions. Section~\ref{S4} analyzes the existence and local stability of equilibria. Section~\ref{S5} is devoted to local bifurcation analysis, while Section~\ref{BT_bif} investigates global bifurcation phenomena including cusp and Bogdanov--Takens bifurcations. Numerical simulations are presented in Section~\ref{S6}, followed by an ecological interpretation of wind effects in Section~\ref{S8}. Finally, Section~\ref{S9} concludes the paper with discussion and future perspectives.





\section{The existence and local stability of the equilibria of the system \eqref{eq:5}}\label{S4}
The existence and local stability of the equilibrium points of the system \eqref{eq:5} are established under the appropriate parametric conditions described in Table \ref{tab:1}.

\section{ Analysis of local bifurcation}\label{S5}

\subsection{Transcritical bifurcation}

 \begin{theorem}\label{th_TC_2}
 The model \eqref{eq:5} experiences a transcritical bifurcation near the axial/ prey only equilibrium point  $E_2(K,0)$ in $\phi(w)=\phi_2^{(TC)}(w)=\frac{d(1+\alpha\xi+bk+k^2)}{\beta(k+\xi)}$ with $\phi(w) \neq \dfrac{rc\xi(1+\alpha\xi+bk+k^2)^3}{\beta k(-1-\alpha\xi+b\xi+2k\xi+k^2)}$ and $(-1-\alpha\xi+b\xi+2k\xi+k^2)\neq0$.   
  \end{theorem}
\begin{proof}
     A transcritical bifurcation occurs between the coexistence/interior equilibrium point $E^*$ and the prey only/axial equilibrium point $E_2$ with respect to the bifurcation parameter $\phi(w)$ and these two equilibrium points exchange their stability at that point $\phi(w)=\phi_2^{(TC)}(w)$.  Now we calculate the Jacobian matrix $ J_2$ in $E_2$ of model \eqref{eq:5}, that is,\\
     $$ J_2 = \left(
\begin{array}{cc}
-r&-\dfrac{\phi(w)k}{1+\alpha \xi + bk + k^2}\\
0&-d+\dfrac{\beta\phi(w)(k+\xi)}{1+\alpha \xi + bk + k^2}
\end{array}
\right).$$\\
From the above it is clear that one eigenvalue of the above Jacobian matrix is negative and the other zero iff $-d+\frac{\beta\phi(w)(k+\xi)}{1+\alpha \xi + bk + k^2}=0$, that is, iff $\phi(w)=\dfrac{d(1+\alpha\xi+bk+k^2)}{\beta(k+\xi)}$, which gives $\phi(w)=\phi_2^{TC}(w)=\dfrac{d(1+\alpha\xi+bk+k^2)}{\beta(k+\xi)}$. We obtain $W=(0,1)^{T}$ and $V=\bigg(\dfrac{-\phi(w)k}{r(1+\alpha\xi+bk+k^2)},1\bigg)^{T}$, where W and V represent the eigenvectors related to the zero eigenvalue of the matrices $ {J_2}^T$ and $ J_2$, respectively.\\
Now,
 \begin{equation*}
 \begin{split}
     &W^Tf_{\phi(w)}(E_2,\phi_2^{(TC)}(w)) =0,\\
    &W^T[Df_{\phi(w)}(E_X,\phi_2^{(TC)}(w))V] =\dfrac{\beta(k+\xi)}{1+\alpha\xi+bk+k^2 } \neq 0,\\
    &W^T[D^2f(E_X,\phi_2^{(TC)}(w))(V,V)]= -2c\xi\phi(w)-\dfrac{2\beta k\phi(w)^2(1+\alpha\xi-b\xi-2k\xi-k^2)}{r(1+\alpha\xi+bk+k^2)^3} \neq 0,\\
    &that \hspace{2mm} is \hspace{2mm} iff \hspace{2mm} \phi(w) \neq \dfrac{rc\xi(1+\alpha\xi+bk+k^2)^3}{\beta k(-1-\alpha\xi+b\xi+2k\xi+k^2)} \hspace{2mm} and \hspace{2mm} (-1-\alpha\xi+b\xi+2k\xi+k^2)\neq0.
\end{split}
 \end{equation*}\\
    Now we can say that the system\eqref{eq:5} experiences a transcritical bifurcation around $\phi(w)$ at $\phi(w)=\phi_2^{TC}(w)$ according to Sotomayor's Theorem (\cite{perko2013differential}).   
    
\end{proof}
\begin{theorem}\label{th_TC_1}
 The model \eqref{eq:5} transcritical bifurcation is experienced near the axial/ prey only equilibrium point  $E_1(0,y_1)$ at $\phi(w)=\phi_1^{(TC)}(w)=\frac{(1+\alpha\xi)(d+rc\xi(1+\alpha\xi))}{\beta\xi}$ with $1+\alpha\xi-b\xi \neq0$ and $\dfrac{\beta\phi(w)\xi}{1+\alpha\xi}>d$ and $ \phi(w) \neq \dfrac{(1+\alpha\xi)[bdk-rc\xi(1+\alpha\xi)^2]}{\beta k(1+\alpha\xi-2b\xi)}$ and $(1+\alpha\xi-2b\xi)\neq0$.   
  \end{theorem}
\begin{proof}
    Similar to the theorem \ref{th_TC_1}.
\end{proof}

\subsection{Saddle-node bifurcation}
Now, the system\eqref{eq:5}  can be represented as,
\begin{equation}\label{eq:19}
    \begin{cases}
        &\dfrac{dx}{dt}= rx\left(1-\frac{x}{k}\right) - \dfrac{\phi(w) x y}{1+\alpha \xi + bx + x^2}=p(x)\bigg(F(x)-y\bigg),\\
        &\dfrac{dy}{dt}= \dfrac{\beta \phi(w) (x+\xi)y}{1+\alpha \xi + bx + x^2} - dy -c\xi \phi(w) y^2=y\bigg(h(x)-c\xi\phi(w)y\bigg),\\
        &x(0)>0, \hspace{2mm} y(0)>0.
    \end{cases}
\end{equation}
where $p(x)=\dfrac{\phi(w)x}{1+\alpha\xi+bx+x^2}$ , $F(x)=\dfrac{r}{\phi(w)}\bigg(1-\dfrac{x}{k}\bigg)(1+\alpha\xi+bx+x^2)$ , $h(x)=\dfrac{\beta\phi(w)(x+\xi)}{1+\alpha\xi+bx+x^2}-d$.\\\\

The corresponding Jacobian matrix is\\
$$ J = \left(
\begin{array}{cc}
p'(x)(F(x)-y)+F'(x)p(x)& -p(x)\\
h'(x)y& (h(x)-c\xi\phi(w)y)-c\xi\phi(w)y\\
\end{array}
\right).$$\\
For interior $E^*=(x^*,y^*)$ with $y^*=F(x^*)$, we have that the Jacobian matrix becomes \\
$$ J^* = \left(
\begin{array}{cc}
F'(x^*)p(x^*)& -p(x^*)\\
h'(x^*)F(x^*)& -c\xi\phi(w)F(x^*)\\
\end{array}
\right).$$\\
Hence, 
\begin{equation}\label{eq:20}
tr(J^*)=F'(x^*)p(x^*) -c\xi\phi(w)F(x^*),
\end{equation}
\begin{equation}\label{eq:21}
    det(J^*)=p(x^*)F(x^*)\bigg(h'(x^*)-c\xi\phi(w)F'(x^*)\bigg).
\end{equation}\\\\
\begin{theorem}\label{th_SN}
   The model \eqref{eq:5} has a saddle-node bifurcation at interior equilibrium $E^*(x^*,y^*)$ when  $\phi(w)=\phi^{SN}_*(w)=\dfrac{rc\xi}{k}\bigg[\dfrac{c\xi(b+2x^*)y^*}{(1+\alpha\xi+bx^*+x^{*2})^2}-\dfrac{\beta(1+\alpha\xi-b\xi-2\xi x^*-x^{*2})}{(1+\alpha\xi+bx^*+x^{*2})^3}\bigg]^{-1} $  and  \\ $\phi(w)\neq-\dfrac{rx^*}{ky^*}\bigg[c\xi-\dfrac{x^*(b+2x^*)}{(1+\alpha\xi+bx^*+x^{*2})^2} \bigg]^{-1} $, $(x^*+\xi)(b+2 x^*)^2\neq(1+\alpha\xi+bx^*+x^{*2})(b+3 x^*+\xi)$ and 
   $\phi(w)\neq \dfrac{rc\xi(k-b-3x^*)(1+\alpha\xi+bx^*+x^{*2})^3}{\beta k[(x^*+\xi)(b+2x^*)^2-(1+\alpha\xi+bx^*+x^{*2})(b+3x^*+\xi)]}$.
\end{theorem}
\begin{proof}
For the system \eqref{eq:19} to  undergo a saddle-node bifurcation at the interior equilibrium point $E^*(x^*,y^*)$, the Jacobian matrix $J^*$ at $(x^*,y^*)$ must have one zero eigenvalue and one nonzero (positive or negative) eigenvalue. Therefore, $det(J^*)=0$ and $tr(J^*)\neq0$.\\

Now we obtain $V=(1, F'(x^*))^T$ and $W=(c\xi\phi(w)F(x^*),-p(x^*))^T$, where V and W represent the eigenvectors related to the zero eigenvalue of the matrices $J^*$ and $J^{*T}$ respectively.\\
Now\\
\begin{align*}
    &W^Tf_{\phi(w)}(x^*,\phi^{SN}_*(w))=-c\xi p(x^*)F(x^*)^2-\dfrac{dp(x^*)F(x^*)}{\phi^{SN}_*(w)}\neq0,\\
    &W^T[Df^2(x^*,\phi^{SN}_*(w))(V,V)]=p(x^*)F(x^*)\bigg(c\xi \phi^{SN}_*(w)F''(x^*)-h''(x^*)\bigg)\neq0 \hspace{2mm}iff\hspace{2mm}c\xi \phi^{SN}_*(w)F''(x^*)\neq h''(x^*),
\end{align*}
that is iff $\phi^{SN}_*(w)\neq \dfrac{rc\xi(k-b-3x^*)(1+\alpha\xi+bx^*+x^{*2})^3}{\beta k[(x^*+\xi)(b+2x^*)^2-(1+\alpha\xi+bx^*+x^{*2})(b+3x^*+\xi)]}$ and \\
$(x^*+\xi)(b+2 x^*)^2\neq(1+\alpha\xi+bx^*+x^{*2})(b+3 x^*+\xi).$\\ 
 As it satisfies the condition of Sotomayor's Theorem (\cite{perko2013differential}), system \eqref{eq:19} experiences a saddle-node bifurcation around $\phi(w)$ at $\phi(w)=\phi_*^{SN}(w)$.
 \end{proof}

\subsection{Hopf bifurcation}
For the system \eqref{eq:19} to undergo a Hopf bifurcation at the interior point $E^*=(x^*,y^*)$, the Jacobian $J^*$ at $(x^*, y^*)$ must have a pair of purely imaginary eigenvalues $\pm i\delta_0$. This implies that the determinant of the Jacobian matrix  $J^*$  is positive  and the trace of the Jacobian matrix $J^*$ vanishes.\\
That is, $det(J^*)=\delta_0^2>0$ and $tr(J^*)=0$.\\
\subsubsection{Direction of Hopf bifurcation and stability of bifurcating periodic solutions}     
 
    Under the condition of Hopf bifurcation, $tr(J^*)=0$ and $det(J^*)>0$, which implies as follows:\\
    \begin{equation}\label{eq:22}
        F'(x^*)p(x^*)=c\xi\phi(w)F(x^*).
    \end{equation}
     We begin by shifting coordinates via the affine transformation $X_1 = x-x^*$ and $Y_1 =y-F(x^*)$, which brings equilibrium $E^* =(x^*,F(x^*))$. Expanding the system in a Taylor series around $E^*$, we obtain the following reduced system:
    \begin{equation}\label{eq:25}
    \begin{cases}
     \dot{X_1} = p(x^*)F'(x^*)X_1-p(x^*)Y_1+\bigg(p'(x^*)F'(x^*)+\dfrac{p(x^*)F''(x^*)}{2}\bigg)X_1^{2}-p'(x^*)X_1Y_1+\\\bigg(\dfrac{F'''(x^*)p(x^*)}{6}+\dfrac{F''(x^*)p'(x^*)}{2}+\dfrac{F'(x^*)p''(x^*)}{2}\bigg)X_1^3-\dfrac{p''(x^*)}{2}X_1^2Y_1+O(|X_1,Y_1|^3),\\
    \dot{Y_1}=F(x^*)h'(x^*)X_1-c\xi\phi(w)F(x^*)Y_1+ \dfrac{F(x^*)h''(x^*)}{2}X_1^2+h'(x^*)X_1Y_1-c\xi\phi(w)Y_1^2 +\\ \dfrac{F(x^*)h'''(x^*)}{6}X_1^3+\dfrac{h''(x^*)}{2}X_1^2Y_1 + O(|X_1,Y_1|^3).
    \end{cases}
    \end{equation}
   Using the condition of Hopf bifurcation \eqref{eq:22}, the system \eqref{eq:19} becomes \\
   \begin{equation}\label{eq:23}
      \begin{cases}
          \dot{X_1}=a_{10}X_1+a_{01}Y_1+a_{20}X_1^2+a_{11}X_1Y_1+a_{30}X_1^3+a_{21}X_1^2Y_1+O(|X_1,Y_1|^4)\\
          \dot{Y_1}= b_{10}X_1-a_{10}Y_1+b_{20}X_1^2+b_{02}Y_1^2+b_{11}X_1Y_1+b_{30}X_1^3+b_{21}X_1^2Y_1+O(|X_1,Y_1|^4),
      \end{cases}
   \end{equation}
  where the coefficients are given by\\
 $a_{10}=p(x^*)F'(x^*)$,  $a_{01}=-p(x^*)$, 
 $a_{20}=p'(x^*)F'(x^*)+\dfrac{p(x^*)F''(x^*)}{2}$, $a_{11}=-p'(x^*)$,\\ $a_{30}=\dfrac{F'''(x^*)p(x^*)}{6}+\dfrac{F''(x^*)p'(x^*)}{2}+\dfrac{F'(x^*)p''(x^*)}{2},$ $a_{21}=-\dfrac{p''(x^*)}{2}$,
 $b_{10}=F(x^*)h'(x^*),$ \\ $b_{20}=\dfrac{F(x^*)h''(x^*)}{2}$,  $b_{02}=-c\xi\phi(w)$,  $b_{11}=h'(x^*)$,$b_{30}=\dfrac{F(x^*)h'''(x^*)}{6}$, $b_{21}=\dfrac{h''(x^*)}{2}$.\\
Hence, we can write \eqref{eq:23} as follows:\\
$$  \left(
\begin{array}{cc}
\dot{X_1} \\
\dot{Y_1}
\end{array}
\right)=J^*(x^*,y^*)\left(
\begin{array}{cc}
{X_1} \\
{Y_1}
\end{array}
\right)+\left(
\begin{array}{cc}
 a_{20}X_1^2+a_{11}X_1Y_1+a_{30}X_1^3+a_{21}X_1^2Y_1+O(|X_1,Y_1|^4)\\
b_{20}X_1^2+b_{02}Y_1^2+b_{11}X_1Y_1+b_{30}X_1^3+b_{21}X_1^2Y_1+O(|X_1,Y_1|^4)
\end{array}
\right).$$\\
Now, to reduce  system \eqref{eq:23} to a more canonical form, we apply the transformation\\
$X_1=X_2$ and $Y_1=\dfrac{-a_{10}}{a_{01}}X_2-\dfrac{\delta_0}{a_{01}}Y_2$, (\cite{li2025turing}).\\
Under this transformation, the above system becomes
\begin{equation}\label{eq:24}
    \begin{cases}
        \dot{X_2}=-\delta_0Y_2+A_{20}X_2^2+A_{11}X_2Y_2+A_{30}X_2^3+A_{21}X_2^2Y_2+O(|X_2,Y_2|^4),\\
        \dot{Y_2}=\delta_0X_2+B_{20}X_2^2+B_{11}X_2Y_2+B_{02}Y_2^2+B_{30}X_2^3+B_{21}X_2^2Y_2+O(|X_2,Y_2|^4),
    \end{cases}
\end{equation}
where the coefficients are:\\
$A_{20}=\bigg(a_{20}-\dfrac{a_{11}a_{10}}{a_{01}}\bigg)$, $A_{11}=-\dfrac{a_{11}\delta_0}{a_{01}}$, $A_{30}=\bigg(a_{30}-\dfrac{a_{21}a_{10}}{a_{01}}\bigg)$, $A_{21}=-\dfrac{a_{21}\delta_0}{a_{01}}$,\\
$B_{20}=-\dfrac{1}{\delta_0}\bigg(b_{20}a_{01}-b_{11}a_{10}+a_{20}a_{10}+\dfrac{b_{02}a_{10}^2}{a_{01}}-\dfrac{a_{11}a_{10}^{2}}{a_{01}}\bigg)$, $B_{11}=\bigg(\dfrac{a_{10}a_{11}}{a_{01}}+b_{11}-2\dfrac{b_{02}a_{10}}{a_{01}}\bigg)$, $B_{02}=-\dfrac{b_{20}\delta_0}{a_{01}}$,\\$B_{30}=\dfrac{a_{10}}{\delta_0}\bigg(\dfrac{a_{20}a_{10}}{a_{01}}-a_{30}-\dfrac{b_{30}a_{01}}{a_{10}} \bigg)$, $B_{21}=\bigg(\dfrac{a_{21}a_{10}}{a_{01}}+b_{21} \bigg).$\\

According to \cite{perko2013differential}[p-353], for  system \eqref{eq:24}, the Lyapunov number is given by \\
$ \alpha_1=\dfrac{3\pi}{2\delta_0^2}\bigg[3\delta_0A_{30}+\delta_0B_{21}+A_{11}A_{20}-B_{11}(B_{20}+B_{02})-2A_{20}B_{20}\bigg]$. 
 
 The direction of the Hopf bifurcation is determined by the sign of the first Lyapunov coefficient $\alpha_1<0$, which yields a supercritical Hopf with stable small-amplitude periodic solutions, while $\alpha_1>0$ gives a subcritical Hopf with unstable cycles. Thus, the stability of the bifurcating periodic solution directly depends on $\alpha_1$, with supercritical Hopf yielding stable limit cycles and subcritical Hopf yielding unstable ones.\\
 Numerically, it is shown that, taking ($b,c,d,r, k, \alpha, \beta, \xi$) = (0.5,0.1,0.1,0.4985,2.5,0.5,0.285,1.5) and for some $\phi(w)$, $\alpha_1>0$.\\

  Hence, the system \eqref{eq:5} exhibits the Hopf Bifurcation at $E^*(x^*,y^*)$, which is subcritical as  $\alpha_1>0$ (fig.-\ref{Hopf_bifurcation}).

 \section{Global bifurcation: Cusp singularity and Bogdanov-Takens bifurcations}\label{BT_bif}

For the system \eqref{eq:19} to undergo a Bogdanov-Takens bifurcation at the interior point $E^*=(x^*,y^*)$, the Jacobian $J^*$ at $(x^*, y^*)$ must have a double zero eigenvalue in a single Jordan
block. This implies that both the determinant and the trace of the Jacobian matrix $J^*$ vanish
simultaneously. \\
That is,
$det (J^*) = 0$ and $tr (J^*) = 0$.\\
Now \begin{equation}\label{tr=0}
tr(J^*) =0  \implies  \phi(w)=\dfrac{r x^*}{k y^*}\bigg(\dfrac{x^*(b+2x^*)}{(1+\alpha \xi+bx^*+x^{*2})^2}-c\xi \bigg)^{-1}
\end{equation} and \\
\begin{equation}\label{det=0}
det(J^*)=0  \implies  \dfrac{c\xi}{k}\bigg[1+\dfrac{(x^*-k)(b+2 x^*)}{(1+\alpha\xi+bx^*+x^{*2})} \bigg]+\dfrac{\beta(k-x^*)(1+\alpha\xi-b\xi-2 x^*\xi-x^{*2})}{k y^*(1+\alpha\xi+bx^*+x^{*2})^2}=0.
\end{equation}
\begin{theorem}
    If $(\xi,\phi(w))=(\xi^{CP},\phi^{CP}(w))$, then the positive equilibrium $E^*$ =$ (x^*, F(x^*))$ is a cusp  of co-dimension-2. 
\end{theorem}
\begin{proof}
    We begin by shifting coordinates via the affine transformation $X_1 = x-x^*$ and $Y_1 =y-F(x^*)$, which brings the equilibrium $E^*$ = $(x^*,F(x^*))$. Expanding the system in a Taylor series around $E^*$, we obtain the following reduced system:
    \begin{equation}\label{eq:25}
    \begin{cases}
     \dot{X_1} = p(x^*)F'(x^*)X_1-p(x^*)Y_1+\bigg(p'(x^*)F'(x^*)+\dfrac{p(x^*)F''(x^*)}{2}\bigg)X_1^{2}-p'(x^*)X_1Y_1+O(|X_1,Y_1|^3),\\
    \dot{Y_1}=F(x^*)h'(x^*)X_1-c\xi\phi(w)F(x^*)Y_1+ \dfrac{F(x^*)h''(x^*)}{2}X_1^2+h'(x^*)X_1Y_1-c\xi\phi(w)Y_1^2 + O(|X_1,Y_1|^3).
    \end{cases}
    \end{equation}
    Under the condition of double zero eigenvalues of $J^*$,\\
    $tr(J^*)=0$ and $det(J^*)=0$, which implies as follows:\\
    \begin{equation}\label{eq:26}
        F'(x^*)p(x^*)=c\xi\phi(w)F(x^*),
    \end{equation}
    \begin{equation}\label{eq:27}
        h'(x^*)=c\xi\phi(w)F'(x^*).
    \end{equation}
Using the condition \eqref{eq:26}, system \eqref{eq:25} becomes \\
   \begin{equation}\label{eq:28}
      \begin{cases}
          \dot{X_1}=a_{10}X_1+a_{01}Y_1+a_{20}X_1^2+a_{11}X_1Y_1+O(|X_1,Y_1|^3),\\
          \dot{Y_1}= b_{10}X_1-a_{10}Y_1+b_{20}X_1^2+b_{02}Y_1^2+b_{11}X_1Y_1+O(|X_1,Y_1|^3),
      \end{cases}
   \end{equation}
   where the coefficients are given by\\
 $a_{10}=p(x^*)F'(x^*)$,  $a_{01}=-p(x^*)$, 
 $a_{20}=p'(x^*)F'(x^*)+\dfrac{p(x^*)F''(x^*)}{2}$, $a_{11}=-p'(x^*)$, \\
 $b_{10}=F(x^*)h'(x^*)$, $b_{20}=\dfrac{F(x^*)h''(x^*)}{2}$,  $b_{02}=-c\xi\phi(w)$, $b_{11}=h'(x^*)$. \\
 From the equation \eqref{eq:26} \\
 $F'(x^*)p(x^*)=c\xi\phi(w)F(x^*)$ $>0$, as $F(x^*)>0$. Hence $F'(x^*)p(x^*)\neq0$ implies $a_{10}\neq0$.\\
 From the equation \eqref{eq:27}\\
  $h'(x^*)=c\xi\phi(w)F'(x^*)\neq0$.Therefore $b_{10}=F(x^*)h'(x^*)\neq0$.\\
Now to reduce the system \eqref{eq:28} to more general canonical form, we apply the transformation\\
$X_1=X_2+\dfrac{Y_2}{b_{10}}$ and $Y_1=\dfrac{Y_2}{a_{10}}$.\\
Under this transformation, the system \eqref{eq:28} becomes \\
\begin{equation}\label{eq:29}
    \begin{cases}
        
        \dot{X_2}=d_{01}Y_2+ d_{20}X_2^2+d_{02}Y_2^2+d_{11}X_2Y_2+O(|X_2,Y_2|^3),\\
        \dot{Y_2}= c_{10}X_2+c_{20}X_2^2+c_{02}Y_2^2+c_{11}X_2Y_2+O(|X_2,Y_2|^3),\\
    \end{cases}
\end{equation}
where the coefficients are as follows:\\
$c_{10}=a_{10}b_{10}$, $c_{20}=a_{10}b_{20}$, $c_{02}=\bigg(\dfrac{a_{10}b_{20}}{b_{10}^2}+\dfrac{b_{02}}{a_{10}}+\dfrac{b_{11}}{b_{10}}\bigg)$, $c_{11}=\bigg(\dfrac{2a_{10}b_{20}}{b_{10}}+b_{11}\bigg)$\\
$d_{01}=\bigg(\dfrac{a_{10}}{b_{10}}+\dfrac{a_{01}}{a_{10}}\bigg)$, $d_{20}=\bigg(a_{20}-\dfrac{a_{10}b_{20}}{b_{10}}\bigg)$, $d_{02}=\bigg(\dfrac{a_{11}}{a_{10}b_{10}}+\dfrac{a_{20}}{b_{10}^2}-\dfrac{a_{10}b_{20}}{b_{10}^3}-\dfrac{b_{02}}{a_{10}b_{10}}-\dfrac{b_{11}}{b_{10}^2}\bigg)$,\\
$d_{11}=\bigg(\dfrac{a_{11}}{a_{10}}+\dfrac{2a_{20}}{b_{10}}-\dfrac{2a_{10}b_{20}}{b_{10}^2}-\dfrac{b_{11}}{b_{10}}\bigg)$.\\\\
Now, since $det(J^*)=0$, implies that $a_{10}^2+b_{10}a_{01}=0$.
That is, $d_{01}=\bigg(\dfrac{a_{10}}{b_{10}}+\dfrac{a_{01}}{a_{10}}\bigg)=0.$\\\\

Then system\eqref{eq:29} becomes
\begin{equation}\label{eq:30}
    \begin{cases}
    \dot{X_2}=d_{20}X_2^2+d_{02}Y_2^2+d_{11}X_2Y_2+O(|X_2,Y_2|^3),\\
        \dot{Y_2}= c_{10}X_2+c_{20}X_2^2+c_{02}Y_2^2+c_{11}X_2Y_2+O(|X_2,Y_2|^3).
    \end{cases}
\end{equation}
Making a $C_{\infty}$-change of variables $X_3=c_{10}X_2+c_{02}Y_2^2-d_{20}X_2Y_2$, $Y_3=Y_2-\dfrac{c_{11}+d_{20}}{2c_{10}}Y_2^2-\dfrac{C_{20}}{C_{10}} X_2Y_2$ in a small neighborhood of $(0, 0)$, system \eqref{eq:30} transforms to  the Standard normal form
\begin{equation}\label{eq:31}
    \begin{cases}
    \dot{X_3}=D_1Y_3^2+D_2X_3Y_3+O(|X_3,Y_3|^3),   \\
   \dot{Y_3}=X_3+O(|X_3,Y_3|^3), \\
    \end{cases}
\end{equation}
where $D_1=c_{10}d_{02}$, $D_2=(d_{11}+2c_{02})$.\\
Notice that, for the values of $\phi^{CP}(w)$ and $\xi^{CP}$, $D_1D_2\neq0$, where $\phi^{CP}(w)$ and $\xi^{CP}$ satisfy equations \eqref{eq:26} and \eqref{eq:27}.\\
Hence the positive equilibrium $E^*$ is a cusp of co-dimension 2.
\end{proof}

We now discuss if the system \eqref{eq:19} undergoes Bogdanov-Takens bifurcation in a small neighborhood of $(\phi(w),\xi)$,
choosing $\phi(w)$ and $\xi$ as bifurcation parameters, we perform a bifurcation analysis of the system \eqref{eq:19} as $(\phi(w),\xi)$ varies near $(\phi^*(w),\xi^*)$, where $\phi^*(w)$ and $\xi^*$ satisfies the equations \eqref{tr=0} and \eqref{det=0}. We consider the following unfolding system of the system \eqref{eq:19}\\ 
\begin{equation}\label{eq:32}
 \begin{cases}
        &\dfrac{dx}{dt}= rx\left(1-\frac{x}{k}\right) - \dfrac{(\phi(w)+\lambda_1) x y}{1+\alpha \xi + bx + x^2}=p_1(x)\bigg(F_1(x)-y\bigg),\\
        &\dfrac{dy}{dt}= \dfrac{\beta \phi(w) (x+\xi)y}{1+\alpha \xi + bx + x^2} - dy -c(\xi+\lambda_2) \phi(w) y^2=y\bigg(h(x)-c(\xi+\lambda_2)\phi(w)y\bigg), \\
        &x(0)>0, \hspace{2mm} y(0)>0.
 \end{cases}
\end{equation}

where $\lambda=(\lambda_1,\lambda_2)$ are small parameters and  $p_1(x)=\dfrac{(\phi(w)+\lambda_1)x}{1+\alpha\xi+bx+x^2}$, \\ $F_1(x)=\dfrac{r}{(\phi(w)+\lambda_1)}\bigg(1-\dfrac{x}{k}\bigg)(1+\alpha\xi+bx+x^2)$, $h(x)=\dfrac{\beta\phi(w)(x+\xi)}{1+\alpha\xi+bx+x^2}-d$.\\

Now we have the following theorem:\\
\begin{theorem}
   If we vary $(\lambda_1,\lambda_2)$ in a small neighborhood of origin, then the system \eqref{eq:19} undergoes Bogdanov-Takens bifurcation in a small neighborhood of $E^*$. 
\end{theorem}
\begin{proof}
    When $\lambda_1=\lambda_2=0$, the equilibrium $E^*$ is a cusp of dimension 2.
    Now we consider Bogdanov-Takens bifurcation of the system \eqref{eq:32} in a small neighborhood of $E^*(x^*,y^*)$. When $(\lambda_1,\lambda_2)$ vary in a small neighborhood of origin, let $x_1=x-x^*$ and $y_1=y-y^*.$ Then the system \eqref{eq:32} becomes\\
\begin{equation}\label{eq:33}
   \begin{cases}
      \dot{x_1}=u_{10}x_1+u_{01}y_1+u_{11}x_1y_1+u_{20}x_1^2+O(|x_1,y_1|^3),\\
      \dot{y_1}=v_{00}+v_{10}x_1+v_{01}y_1+v_{11}x_1y_1+v_{20}x_1^2+v_{02}y_1^2+O(|x_1,y_1|^3),
      \end{cases}
\end{equation}
   where $u_{10}=p_1(x^*)F_1'(x^*)$  , $u_{01}=-p_1(x^*)$, 
 $u_{20}=p_1'(x^*)F_1'(x^*)+\dfrac{p_1(x^*)F_1''(x^*)}{2}$, $u_{11}=-p_1'(x^*)$ \\\\
$v_{00}=-c\lambda_2\phi(w)F_1(x^*)$, $v_{10}=F_1(x^*)h'(x^*)$ , $v_{01}=-c(\xi +\lambda_2)\phi(w)F_1(x^*)-c\lambda_2\phi(w)F_1(x^*)$ ,\\\\ $v_{20}=\dfrac{F_1(x^*)h''(x^*)}{2}$,  $v_{02}=-c(\xi +\lambda_2)\phi(w)$, $v_{11}=h'(x^*)$. \\
We now apply the transformation $x_2=x_1$ and $y_2=u_{10}x_1+u_{01}y_1$ as in \cite{gong2014bogdanov,huang2014bifurcations}. The system \eqref{eq:33} becomes:\\
\begin{equation}\label{eq:34}
    \begin{cases}
    \dot{x_2}=&y_2+\bigg(u_{20}-\dfrac{u_{11}u_{10}}{u_{01}}\bigg)x_2^2+\dfrac{u_{11}}{u_{01}}x_2y_2+O(|x_2,y_2|^3),\\
       \dot{ y_2}=&u_{01}v_{00}+(u_{01}v_{10}-u_{10}v_{01})x_2+(u_{10}+v_{01})y_2+\bigg(u_{10}u_{20}-\dfrac{u_{11}u_{10}^2}{u_{01}}+v_{20}u_{01}+\dfrac{v_{02}u_{10}^2}{u_{01}}-v_{11}u_{10} \bigg)x_2^2+\\&\bigg(\dfrac{u_{11}u_{10}}{u_{01}}+v_{11}-2v_{02}u_{10} \bigg)x_2y_2+\dfrac{v_{02}}{u_{01}}y_2^2+O(|x_2,y_2|^3).
    \end{cases}
\end{equation}
Now, we using the following $C^{\infty}$ change of co-ordinates in a small neighborhood of (0,0) as \\
\begin{equation*}
   \begin{split}
   & x_3=x_2-\bigg(\dfrac{u_{11}}{2u_{01}}+\dfrac{v_{02}}{2u_{01}} \bigg)x_2^2,\\
    &y_3=y_2+\bigg(u_{20}-\dfrac{u_{11}u_{10}}{u_{01}} \bigg)x_2^2-\dfrac{v_{02}}{u_{01}}x_2y_2.
    \end{split}
\end{equation*}
Hence, the system \eqref{eq:34} becomes\\
\begin{equation}\label{eq:35}
    \begin{cases}
    \dot{x_3}=&y_3+O(|x_3,y_3|^3),\\
    \dot{y_3}=&u_{01}v_{00}+(u_{01}v_{10}-u_{10}v_{01}-v_{02}v_{00})x_3+(u_{10}+v_{01})y_3+\bigg(v_{11}-2v_{02}u_{10}+2u_{20}-\dfrac{u_{11}u_{10}}{u_{01}}\bigg)x_3y_3\\&+A_1(\lambda)x_3^2+O(|x_3,y_3|^3),
    \end{cases}
\end{equation}
where $A_1(\lambda)=u_{10}u_{20}-\dfrac{u_{11}u_{10}^2}{u_{01}}+v_{20}u_{01}+\dfrac{v_{02}u_{10}^2}{u_{01}}-v_{11}u_{10}-\dfrac{v_{02}}{u_{01}}(u_{01}v_{10}-u_{10}v_{01})-(u_{10}+v_{01})\bigg(u_{20}-\dfrac{u_{11}u_{10}}{u_{01}}\bigg)+\bigg(\dfrac{u_{11}}{2u_{01}}+\dfrac{v_{02}}{2u_{01}} \bigg) \bigg(u_{01}v_{10}-u_{10}v_{01}-v_{02}v_{00}\bigg)$. \\

Now, we choose another $C^{\infty}$ change of co-ordinates near the origin as
$Z_1=x_3$ and $Z_2=y_3+O(|x_3,y_3|^3)$ and obtained from system \eqref{eq:35} that\\
\begin{equation}\label{eq:36}
    \begin{cases}
        \dot{Z_1}=&Z_2,\\
        \dot{Z_2}=&u_{01}v_{00}+(u_{01}v_{10}-u_{10}v_{01}-v_{02}v_{00})Z_1+(u_{10}+v_{01})Z_2+A_1(\lambda)Z_1^2+G_1(Z_1)+Z_2G_2(Z_1)+ \\&\bigg(v_{11}-2v_{02}u_{10}+2u_{20}-\dfrac{u_{11}u_{10}}{u_{01}}\bigg)Z_1Z_2+Z_2^2G_{3}(Z_1,Z_2),
    \end{cases}
\end{equation}
where $G_1,G_2 $ are $C^{\infty}$ in $Z_1$ and $G_3$ is $C^{\infty}$ in $(Z_1,Z_2)$ and $G_1(Z_1)=O(|Z_1|^3)$, $G_2(Z_1)=O(|Z_1|^2)$, $G_3(Z_1,Z_2)=O(|Z_1,Z_2|)$. We rewrite the system \eqref{eq:36} as follows:\\
\begin{equation}\label{eq:37}
\begin{cases}
    \dot{Z_1}=Z_2,\\
        \dot{Z_2}=\theta(Z_1,\lambda)+(u_{10}+v_{01})Z_2+\bigg(v_{11}-2v_{02}u_{10}+2u_{20}-\dfrac{u_{11}u_{10}}{u_{01}}\bigg)Z_1Z_2+Z_2G_2(Z_1)+Z_2^2G_{3}(Z_1,Z_2),
\end{cases}
\end{equation}
where  $\theta(Z_1,\lambda)=u_{01}v_{00}+(u_{01}v_{10}-u_{10}v_{01}-v_{02}v_{00})Z_1+A_1(\lambda)Z_1^2+G_1(Z_1)$.\\
Now, for certain conditions of $\phi(w)$ and $\xi$, we have \\
$A_1(0)=u_{10}u_{20}-\dfrac{u_{11}u_{10}^2}{u_{01}}+v_{20}u_{01}+\dfrac{v_{02}u_{10}^2}{u_{01}}-v_{11}u_{10}>0$. \\
Applying the Malgrange Preparation Theorem [\cite{chow2012methods,huang2014bifurcations}] to $\theta(Z_1,\lambda)$, we have\\
$\theta(Z_1,\lambda)=\bigg(\eta_1(\lambda)+\eta_2(\lambda )Z_1+Z_1^2 \bigg)\Phi(Z_1,\lambda)$, \\
where $\eta_1(\lambda)=\dfrac{u_{01}v_{00}}{A_1(\lambda)}$,
 $\eta_2(\lambda)=\dfrac{u_{01}v_{10}-u_{10}v_{01}-v_{02}v_{00}}{A_1(\lambda)}$ and $\Phi(0,\lambda)=A_1(\lambda)$ and $\Phi(Z_1,\lambda)$ is a power series in $Z_1$, whose coefficients depend on the parameters $\lambda=(\lambda_1,\lambda_2)$.\\
Now we consider $Q_1=Z_1$ , $Q_2=\dfrac{Z_2}{\sqrt{\Phi(Z_1,\lambda)}}$ and $T=\int _0^t\sqrt(\Phi(Z_1(s),\lambda))\,ds$. Then the system \eqref{eq:37} becomes\\
\begin{equation}\label{eq:38}
    \begin{cases}
        \dot{Q_1}=Q_2,\\
        \dot{Q_2}=\eta_1(\lambda)+\eta_2(\lambda )Q_1+Q_1^2+\eta_3(\lambda)Q_2+\eta_4(\lambda)Q_1Q_2+ R(Q_1,Q_2,\lambda),
    \end{cases}
\end{equation}
where $\eta_3(\lambda)=\dfrac{(u_{10}+v_{01})}{\sqrt{(\Phi(0,\lambda))}}$, $\eta_4(\lambda)=\dfrac{B}{\sqrt{(\Phi(0,\lambda))}}$, $B=v_{11}+2u_{20}-\dfrac{u_{11}u_{10}}{u_{01}}-2v_{02}u_{10}-\dfrac{v_{02}(u_{10}+v_{01})}{u_{01}}$ and $R(Q_1,Q_2,0)$ is a power series of $(Q_1,Q_2)$ with the term $Q_1^i Q_2^j$ such that $i+j\geq 3$, $i\geq 4$ and $j\geq2$.\\
By numerically we have seen that, for some parameter values of $\phi(w)$ and $\xi$, $B>0$. Hence with certain condition $\eta_4(\lambda)>0$.

Now we make an another affine transformation as $X=Q_1+\dfrac{\eta_1(\lambda)}{2}$ and $Y=Q_2.$ Then using Taylor's series expansion, we have the system \eqref{eq:38} as follows: \\
\begin{equation}\label{eq:39}
    \begin{cases}
        \dot{X}=Y,\\
        \dot{Y}=\mu_1(\lambda_1,\lambda_2)+\mu_2(\lambda_1,\lambda_2)Y+X^2+\eta_4(\lambda)XY+O(|X,Y|^3),
    \end{cases}
\end{equation}
where $\mu_1(\lambda_1,\lambda_2)=\eta_1(\lambda)-\dfrac{\eta_2^2(\lambda)}{4}$ and $\mu_2(\lambda_1,\lambda_2)=\bigg(\eta_3(\lambda)-\dfrac{1}{2}\eta_2(\lambda)\eta_4(\lambda)\bigg)$.\\
Now, with parameter values of $\phi^*(w)$ and $\xi^*$, we have\\ $\bigg|\dfrac{\partial(\eta_1,\eta_2)}{\partial(\lambda_1,\lambda_2)} \bigg|_{\lambda_1=\lambda_2=0}=-\dfrac{c^2\xi \phi(w)pF^2 h'}{2 A_1(0)^{5/2}}(h'+p'F'+pF''+2ph')\neq0.$ \\

Using lemma 8.8 in \cite{kuznetsov1998elements}, the system \eqref{eq:39} is locally topologically equivalent to the following\\
\begin{equation}\label{eq:40}
    \begin{cases}
        \dot{X}=Y,\\
        \dot{Y}=\mu_1(\lambda_1,\lambda_2)+\mu_2(\lambda_1,\lambda_2)Y+X^2+\eta_4(\lambda)XY,
    \end{cases}
\end{equation}
which is strongly topologically equivalent to \\
\begin{equation}\label{eq:41}
    \begin{cases}
        \dot{X}=Y,\\
        \dot{Y}=\mu_1(\lambda_1,\lambda_2)+\mu_2(\lambda_1,\lambda_2)Y+X^2+XY.
    \end{cases}
\end{equation}

Choosing $\mu_1, \mu_2$ as bifurcation parameters, the system \eqref{eq:40} undergoes a Bogdanov-Takens bifurcation in a small neighborhood of $E^*$, when $(\lambda_1,\lambda_2)$ vary in a small neighborhood of origin.\\
Now we observe that, if $\mu_1>0$, the system \eqref{eq:40} has no critical value. If $\mu_1=0 $ and $\mu_2\neq0$, then the system \eqref{eq:40} has only one critical value $(0,0)$, which is non hyperbolic. In this case system \eqref{eq:40} has saddle-node at origin.\\

If $\mu_1<0$, then the system \eqref{eq:40} has two equilibrium points $(\sqrt{-\mu_1},0)$ and $(-\sqrt{-\mu_1},0)$.
Now Jacobian matrix for any equilibrium of the system \eqref{eq:40} is given by\\
$$ J' = \left(
\begin{array}{cc}
0& 1\\
2X+\eta_4(\lambda)Y& \mu_2+\eta_4(\lambda)X\\
\end{array}
\right).$$\\
As the eigenvalues for the equilibrium point
$(\sqrt{-\mu_1},0)$ are\\ $\dfrac{1}{2}\bigg[\mu_2+\eta_4(\lambda)\sqrt{-\mu_1}\pm\sqrt{\bigg((\mu_2+\eta_4(\lambda)\sqrt{-\mu_1})^2+8\sqrt{-\mu_1} \bigg)}\bigg]$, then it is saddle point.\\
The eigenvalues of the equilibrium point $(-\sqrt{-\mu_1,0)}$ are \\ $ \dfrac{1}{2}\bigg[\mu_2-\eta_4(\lambda)\sqrt{-\mu_1}\pm\sqrt{\bigg((\mu_2+\eta_4(\lambda)\sqrt{-\mu_1})^2-8\sqrt{-\mu_1} \bigg)}\bigg]$. \\
If $\mu_2>\eta_4(\lambda)\sqrt{-\mu_1}$, then stable focus. \\
If $\mu_2<\eta_4(\lambda)\sqrt{-\mu_1}$, then unstable focus and\\
if $\mu_2=\eta_4(\lambda)\sqrt{-\mu_1}$ then non-hyperbolic.\\
  
We observe that for the system \eqref{eq:40}, Lyapunov number according to \cite{perko2013differential} is given by \\
$\sigma=\dfrac{3\pi }{2(2\sqrt{-\mu_1})^{3/2}}\eta_4(\lambda)>0$. \\
Hence the system \eqref{eq:40} undergoes a subcritical hopf bifurcation at $\mu_2=\eta_4(\lambda)\sqrt{-\mu_1}$.\\

Hence we obtain the local representations of the bifurcation curves as follows: 
 \begin{align*}
(i)\; \text{Saddle--node bifurcation curve:}\quad 
SN = \Big\{(\lambda_1,\lambda_2):\;
&\mu_1(\lambda_1,\lambda_2)=0, \\
&\mu_2(\lambda_1,\lambda_2)\neq 0 \Big\}. \\
(ii)\; \text{Hopf bifurcation curve:}\quad 
H = \Big\{(\lambda_1,\lambda_2):\;
&\mu_1(\lambda_1,\lambda_2)<0, \\
&\mu_2(\lambda_1,\lambda_2)=
\eta_4(\lambda)
\sqrt{-\mu_1(\lambda_1,\lambda_2)} \Big\}. \\[2mm]
\end{align*}
Now from \cite{perko2013differential}[p-481], we obtain \\
\begin{align*}
    (iii)\; \text{Homoclinic bifurcation curve:}\quad 
HC = \Big\{(\lambda_1,\lambda_2):\;
&\mu_1(\lambda_1,\lambda_2)<0, \\
&\mu_2(\lambda_1,\lambda_2)=
\dfrac{5}{7}\eta_4(\lambda)
\sqrt{-\mu_1(\lambda_1,\lambda_2)} \Big\}.
\end{align*}

\end{proof}

\section{Numerical simulations}\label{S6}
We now analyze the numerical simulations for the models \eqref{eq:5} and \eqref{eq:19}. Here we discuss how the wind function performs to find the number of equilibrium points, their stability, and the occurrence of different types of bifurcation.\\
In Figs. \ref{TS_E2}, \ref{TS_E1} and \ref{TS_E*} we have performed a time series analysis over a long time period $( t = 1000)$ to obtain the values of the stable equilibrium points  $E_2(k,0)$, $E_1(0,y_1)$ and $E^*(x^*,y^*)$ from the set of carefully chosen parameters. The set of parameters chosen to obtain a stable equilibrium is  r=0.5, k = 4, $\alpha = 0.4,\hspace{2mm}\xi=0.6,\hspace{2mm}b = 0.1,\hspace{2mm}\beta = 0.6,\hspace{2mm}d = 0.1,\hspace{2mm}c = 0.2$. When $\phi(w)$=1.5  prey-free equilibrium point $E_1$ is stable, when $\phi(w)=1$  interior equilibrium point $E^*$ is stable and for $\phi(w)=0.6$ predator-free equilibrium point $E_2$ is stable. Therefore, the figures show with decreasing $\phi(w)$ the fitness of predator switches to fitness of prey. \\
In Fig. \ref{R_P_1}, the stability dynamics of the region is stated by different colours. The figure interprets the curves $-tr(J^*)=0$, $det(J^*)=0$, $\dfrac{d(1+\alpha \xi + bk + k^2)}{\beta(k+\xi)}-\phi(w)=0$ and\\
$\phi(w)-\dfrac{\bigg(d+rc\xi(1+\alpha\xi)\bigg)(1+\alpha\xi)}{\beta \xi}=0$, which are the stability conditions of the all equilibriums with respect to $\phi(w)$. In the region $R_1$(blue), both equilibria $E^*$ and $E_1$ are stable. Hence in this region the system \eqref{eq:5} experiences bi-stability between $E^*$ and $E_1$. In the region $R_2$(green), only the predator-free equilibrium $E_2$ is stable. For $R_3$(red), both $E_2$ and $E^*$ are stable, which shows that in this region the system \eqref{eq:5} experiences the bi-stability between $E_2 $ and $E^*$. In the region $R_4$(magenta) all the equilibriums are stable, which indicates that  the system \eqref{eq:5} also has bi-stability between two axial equilibrium $E_1$ and $E_2$.  For the region $R_5$(yellow), only the interior equilibrium $E^*$ is stable.\\   
In Figs. \ref{TC_2} and \ref{TC_1}, we get the existence of transcritical bifurcation of the system \eqref{eq:5}. The parameter set chosen to get these two bifurcations is r=0.5, k=2.5, b=0.5, d=0.1, $\xi=1.5,\hspace{2mm}\beta=0.285,\hspace{2mm}\alpha=0.5,\hspace{2mm}c=0.1.$  In fig. \ref{TC_2}, the points $E^*$ and $E_2$ are stable and unstable, respectively, for the condition $\phi(w)>\phi_2^{(TC)}(w)=0.812$, marked by black dot(TC) and $E^*$(which is not feasible here) and $E_2$ are unstable  and stable, respectively, for the condition $\phi(w)<\phi_2^{(TC)}(w)$(in theorem \ref{th_TC_2}). Hence, we get a transcritical bifurcation between $E_2$ and $E^*$. In fig. \ref{TC_1}, the points $E^*$ and $E_1$ are stable and unstable respectively, for the condition $\phi(w)<\phi_1^{(TC)}(w)=0.95$, marked by black dot(TC) and $E^*$(which is infeasible here) and $E_1$ are unstable  and stable respectively, for the condition $\phi(w)>\phi_1^{(TC)}(w)$(in theorem \ref{th_TC_1}). Hence, we get a transcritical bifurcation between $E_1$ and $E^*$. Hence, we get 2 transcritical bifurcations for the same parameter set with different values of $\phi(w)$.\\
In Figs. \ref{SN_1} and \ref{SN_2},  we get the existence of Saddle-Node bifurcation of the system \eqref{eq:5}. The parameter set chosen to get these two bifurcation is r=0.5, k=2.5, b=0.5, d=0.1, $\xi=1.5,$ $ \beta=0.285,$ $ \alpha=0.5$, c=0.1.
In fig. \ref{SN_1}, two interior points (one stable and one unstable) coincide at $\phi(w)=\phi_*^{SN}(w)=1.3985,$ (in theorem \ref{th_SN}) marked by a black dot(SN). When $\phi(w)>\phi_*^{SN}(w)$, there is no interior equilibrium. In fig. \ref{SN_2}, two interior points ( one stable and one unstable) coincide at $\phi(w)=\phi_*^{SN}(w)=0.895,$ marked by a black dot(SN). When $\phi(w)<\phi_*^{SN}(w)$, there is no interior equilibrium. Hence, we get saddle-node bifurcation at  $\phi(w)=1.3985$ and $\phi(w)=0.895$. Hence, two saddle-node bifurcations for the same parameter set with different values of $\phi(w)$.\\  
In fig. \ref{BT_dynamics}, the global bifurcation that can be seen in section \ref{BT_bif} has been studied numerically. The figure \ref{bt_1} shows the codimension bifurcation dynamics, and it is observed that multiple interior equilibrium points cause the two saddle-node bifurcation and a Hopf bifurcation. It is well-known that a region with Hopf and saddle-node bifurcation causes a Bogdanov-Takens (BT) bifurcation \cite{kuznetsov1998elements, beyn1994numerical}. We study the $(\xi,\phi(w))$ parameter space in fig. \ref{bt_2} and show the existence of BT bifurcation and cusp (CP) of codimension 2. Fig. \ref{bt_hopf_curve} shows the two-parameter Hopf curve on which the trace of the Jacobian is always zero. As is established in past studies, the neighborhood of the BT bifurcation shows the existence of a homoclinic orbit. In fig. \ref{homoclinic}, it can be seen that when two interior equilibrium points are close to each other at the BT-bifurcation neighborhood, a homoclinic orbit is formed.\\
In Fig. \ref{fig:1}, we present the plot of $x$ versus
 $f(x)=x^5+Ax^4+Bx^3+Cx^2+Dx+E$ for four different values of the wind
 function $\phi(w)$. The others parameters are  r=0.5, k=2.5, b=0.5, d=0.1, $\xi=1.5,$ $\beta=0.285$, $\alpha=0.5$, c=0.1. For $\phi(w)=0.92$,  $\phi(w)=0.892$,  $\phi(w)=0.86$ and $\phi(w)=1.5$, the expression $f(x)$ cut the positive part of $f(x)=0$ at three points, two points, one point and nowhere respectively. Hence we get three interior, two interior, unique interior and no interior equilibrium respectively, of the system \eqref{eq:5}. \\

\section{Effects of different wind-predation rate interaction on system dynamics}\label{S8}
In this section, we discussed the effects of the wind on population dynamics. Since wind intensity can vary seasonally or stochastically, its inclusion in ecological models helps explain complex dynamics such as instability, coexistence, or sudden regime shifts.
Here we choose four different predation functions $\phi(w)$ with respect to wind in the system \eqref{eq:5} that provides a more realistic framework to understand population dynamics in flowing environments.\\
\subsection{Bounded U-shaped wind effect}
To capture the non–monotonic influence of wind speed on predation, the predation rate is modeled by a U-shaped bounded function $\phi(w)=2\big(1-e^{-6(w-w_0)^2}\big)$, which attains minimal values near the intermediate wind level $w_0$  and increases toward both lower and higher wind extremes. This functional choice is intended to represent ecological situations in which predation pressure is intensified under environmental stress at both extremes of wind intensity, while moderate wind conditions reduce effective predator–prey encounters. The bounded nature of the function ensures biological realism by preventing unbounded growth of predation intensity.

The bifurcation structures corresponding to this formulation are illustrated in Figures \ref{2exp_1} and \ref{2exp_2}. It is observed that the stability of boundary equilibria depends strongly on wind speed. In the regions of low and high wind intensity, the prey-free equilibrium is found to be stable, indicating that strong predation pressure at these extremes can drive the prey population to extinction. In contrast, for intermediate wind speeds, the predator-free equilibrium becomes stable, suggesting reduced predation efficiency and possible predator decline. Figures \ref{2exp_3} and \ref{2exp_4} illustrate the time evolution of the prey and predator populations, respectively. The corresponding phase portraits for different values of the wind speed parameter $w$ are presented in Figure \ref{2exp_5}. The diagram \ref{2exp_6} further reveal parameter intervals where multiple attractors coexist, and as shown in the corresponding phase portraits, bistability is present, meaning that the long-term outcome depends sensitively on initial population densities. 

From a biological perspective, these results highlight that wind acts as a nonlinear environmental regulator capable of generating alternative stable states. Both low and high wind conditions may enhance predator success, whereas intermediate wind speeds may provide a refuge-like zone for prey. Consequently, moderate environmental conditions may stabilize coexistence, while extreme conditions promote dominance of one species, emphasizing the ecological significance of incorporating U-shaped environmental responses into predator–prey models.

\subsection{Periodic wind effect}
To represent recurrent environmental fluctuations in wind intensity, the predation rate is modeled as a periodic function, $\phi(w)=1-sin(2\pi w)$. This formulation is chosen to reflect cyclical variations in wind conditions, such as seasonal or rhythmic atmospheric patterns, which may repeatedly enhance or reduce predator–prey encounter rates. The periodic structure allows the predation pressure to alternate regularly between higher and lower levels as wind speed varies.

The corresponding dynamical behavior is depicted in Figures \ref{sin_1} and \ref{sin_2}. It is observed that within each period of the function, the stability of equilibria alternates systematically. Specifically, for wind speeds near the lower and upper phases of each cycle, the prey-free equilibrium is found to be stable, indicating intensified predation pressure. In contrast, in the intermediate phase of each period, the predator-free equilibrium becomes stable, suggesting reduced predator efficiency and possible predator decline. The bifurcation patterns repeat consistently across successive periods, demonstrating that the qualitative dynamics are periodically modulated by wind speed. As wind speed progresses through successive periods, the system repeatedly shifts between prey-dominated and predator-dominated states, emphasizing the dynamic nature of ecological responses under periodic environmental drivers. Figures \ref{sin_3} and \ref{sin_4} illustrate the time series dynamics of the prey and predator populations, respectively. Figure \ref{sin_5} presents the corresponding phase portraits for different values of the wind speed parameter $w$. These figures collectively demonstrate how variations in wind speed influence the temporal evolution and qualitative behavior of the prey–predator system.

Biologically, these results indicate that cyclical wind fluctuations can generate repeated transitions between predator-dominated and prey-dominated states. Hence, wind speed acts as a time-varying environmental driver capable of inducing regular shifts in ecosystem structure, highlighting the importance of incorporating periodic environmental forcing in prey–predator models.

\subsection{Exponential wind effect}
To understand the influence of wind on predator–prey interactions, an increasing exponential predation function of the form $\phi(w)=\dfrac{e^{\sigma w}}{5}$
is considered. This functional choice reflects the assumption that predation intensity accelerates with increasing wind speed, such that small increments in wind at higher levels produce disproportionately large increases in predation pressure. Such a formulation is ecologically justified when wind enhances predator efficiency through mechanisms such as improved scent dispersion or the flushing of prey from refuges.

The dynamical consequences of this assumption are illustrated in Figures \ref{exp^sigma_1} and \ref{exp^sigma_2}. Variations in wind speed are observed to significantly affect the stability properties of the equilibria. For relatively low wind speeds, the predator-free equilibrium is found to be stable, indicating predator extinction and prey persistence. As the intensity of the wind increases, the system undergoes qualitative changes, and for sufficiently high wind speeds, the equilibrium of prey-free becomes stable, implying prey extinction due to enhanced predation pressure. These transitions suggest that wind modifies dispersal patterns and hunting efficiency in a nontrivial manner, leading to shifts in dominance between species. Furthermore, the figure \ref{exp^sigma_6} indicate the presence of bistability in certain parameter regimes, where multiple stable states coexist and the long-term outcome depends on initial population densities. Figures \ref{exp^sigma_3} and \ref{exp^sigma_4} represent the temporal dynamics of the prey and predator populations, respectively, whereas Figure \ref{exp^sigma_5} shows the phase portraits corresponding to various values of the wind speed parameter $w$.

Biologically, it can therefore be concluded that wind speed acts as a crucial external environmental driver capable of restructuring community outcomes. Depending on its magnitude, wind may either suppress predator populations or intensify predation to the extent of prey collapse. Therefore, climatic factors such as wind should be regarded as critical external forces that can reshape predator–prey dynamics, potentially leading to population imbalance or even local prey depletion under sustained high-wind conditions. The presence of bistability suggests ecological sensitivity to initial conditions and environmental fluctuations, highlighting the potential for abrupt regime shifts in wind-influenced ecosystems.

\subsection{Bounded concave wind effect}
To incorporate a more ecologically realistic wind effect, a concave bounded predation function of the form $\phi(w)=\dfrac{2w}{\sqrt{1+w^4}}$ is considered. This functional response is selected to represent a situation in which predation initially increases with wind speed but gradually saturates, reflecting biological constraints that prevent indefinite growth in hunting efficiency. This formulation captures the idea that while moderate wind may improve encounter rates, excessive wind can limit predator performance due to movement difficulty, sensory interference, or energetic costs.

The bifurcation structure with respect to wind speed is depicted in Figures \ref{square_1} and \ref{square_2}. From the diagrams, it is observed that for both low and high wind speeds, the predator-free equilibrium is stable, indicating predator extinction under these environmental extremes. In contrast, within intermediate ranges of wind intensity, the prey-free equilibrium becomes stable, suggesting that predators successfully dominate when wind conditions are moderately favorable. The transition between these stability regions highlights the presence of wind-driven qualitative changes in the system dynamics. Furthermore, the diagram \ref{square_6} reveal the presence of bistability within certain wind intervals, where two stable equilibria coexist. In these regions, the long-term outcome depends sensitively on initial population densities. Figures \ref{square_3} and \ref{square_4} depict the time series of the prey and predator populations, respectively, while Figure \ref{square_5} illustrates the corresponding phase portraits for different values of the wind speed $w$.

From a biological perspective, these results indicate that wind acts as a regulating environmental factor with a saturating influence on predation. Moderate wind speeds may enhance predator success by disturbing prey movement or reducing escape efficiency, whereas strong wind imposes limitations on both species, preventing unlimited escalation of predation pressure. The emergence of bistability further suggests that ecosystem responses to wind variation can be non-unique and history-dependent, highlighting the complex role of environmental forcing in shaping population persistence and extinction patterns.

\section{Discussion}\label{S9}
In this study, we investigate the dynamical behavior of a prey–predator system incorporating the effect of wind on population interactions. The model exhibits rich dynamical features like stability switching, oscillatory dynamics and global bifurcations. From a biological perspective, this means that population dynamics are not governed solely by intrinsic growth and predation parameters, but also by external abiotic drivers that can regulate or destabilize ecological balance. Unlike classical prey–predator models, where population changes depend solely on biological interactions, the inclusion of wind introduces an external ecological driver that significantly alters system stability and long-term behavior. As a result, wind acts not merely as a perturbation but as a key control parameter governing the qualitative behavior of the system. In natural ecosystems—such as grasslands, coastal regions, agricultural fields, or open marine environments—wind can significantly alter foraging efficiency, prey vulnerability, and habitat accessibility.

The system shows the phenomenon of bistability, which can be biologically crucial. Ecologically, this indicates the presence of alternative stable states in the predator–prey system, where long-term population outcomes depend on initial conditions and disturbance magnitude by wind. Such systems exhibit threshold behavior and reduced resilience, implying that environmental forcing, including wind variability, can trigger abrupt regime shifts between coexistence and oscillatory or extinction states.

From the bifurcation perspective, parameter variation reveals transitions between distinct dynamical regimes. For different wind intensities, the system typically maintains a stable coexistence equilibrium, indicating balanced prey-predator interaction. However, as wind-related parameters vary, the stability of equilibrium points changes qualitatively. The emergence of saddle–node or transcritical bifurcations indicates thresholds at which coexistence states are created or destroyed, reflecting ecological tipping points in which small environmental changes can result in population extinction or persistence.

At low wind intensity, the system undergoes a subcritical Hopf bifurcation, producing unstable limit cycles around the coexistence equilibrium. Ecologically, this indicates that predator–prey populations may experience abrupt transitions and reduced resilience under wind disturbances.

Wind significantly influences prey-predator interactions by altering encounter and capture efficiency under varying environmental conditions. Changes in wind speed can either enhance or reduce interaction strength, reflecting ecological constraints and behavioral responses of both species. Under fluctuating wind conditions, interaction intensity may vary cyclically, leading to recurrent population oscillations rather than stable coexistence. In some situations, small variations in wind can produce rapid changes in interaction strength, making the system highly sensitive to environmental forcing and potentially causing sudden transitions in population dynamics.

Furthermore, under appropriate parameter conditions, we see that the system exhibits richer codimension-two dynamics (ex., Bogdanov-Takens Bifurcation). Near this singularity, the system can exhibit multiple dynamical phenomena, including saddle-node bifurcations, Hopf bifurcations, and homoclinic orbits. Also for different type of predation function, we have seen dynamical changes of population.

Overall, the biological interpretation emphasizes that wind acts as a regulatory environmental driver that can reshape prey-predator stability, persistence, and resilience. The model highlights the ecological importance of incorporating abiotic factors into population dynamics studies, particularly for predicting ecosystem responses to environmental variability and climate-induced changes.
Although the present study provides a comprehensive qualitative analysis of the proposed prey–predator model under the influence of wind-dependent predation, several important directions remain open for future investigation.
First, spatial heterogeneity could be included through reaction–diffusion formulations. Wind often plays a significant role in spatial dispersal of species; therefore, extending the model to a spatial domain may lead to pattern formation, traveling wave solutions, or Turing instability.
The higher order rich bifurcations such as like saddle-transcritical bifurcation, generalized Hopf etc. can be studied further. Experimental results to understand the functional response behavior for particular ecosystems are yet to be analyzed where wind speed over time can be of great importance to study an autonomous system and produce time dependent predictions for the species populations.
\section*{Conflict of interest}
The authors have no conflict of interest.
 \begin{figure}[H]
	\centering
    \subfloat[]{\includegraphics[width=2.2in, height=1.4 in]{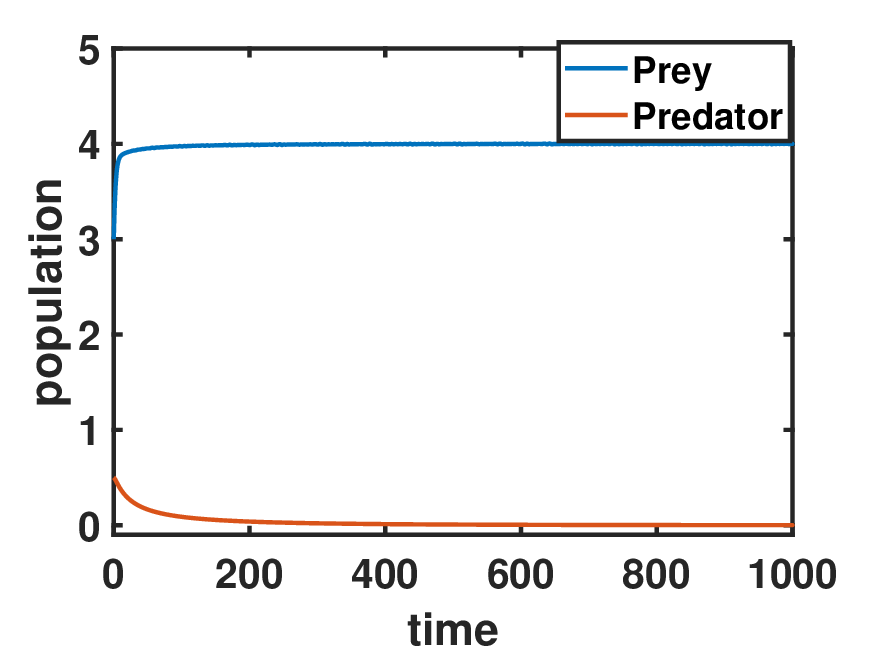}\label{TS_E2}}
	\subfloat[]{\includegraphics[width=2.2 in, height=1.4 in]{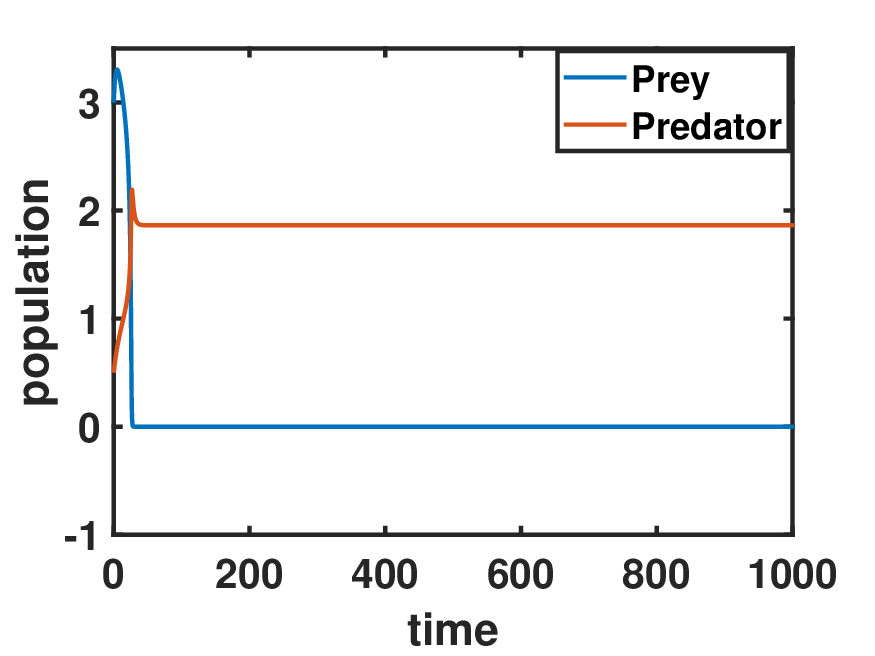}\label{TS_E1}} 	
	\subfloat[]{\includegraphics[width=2.2 in, height=1.4 in]{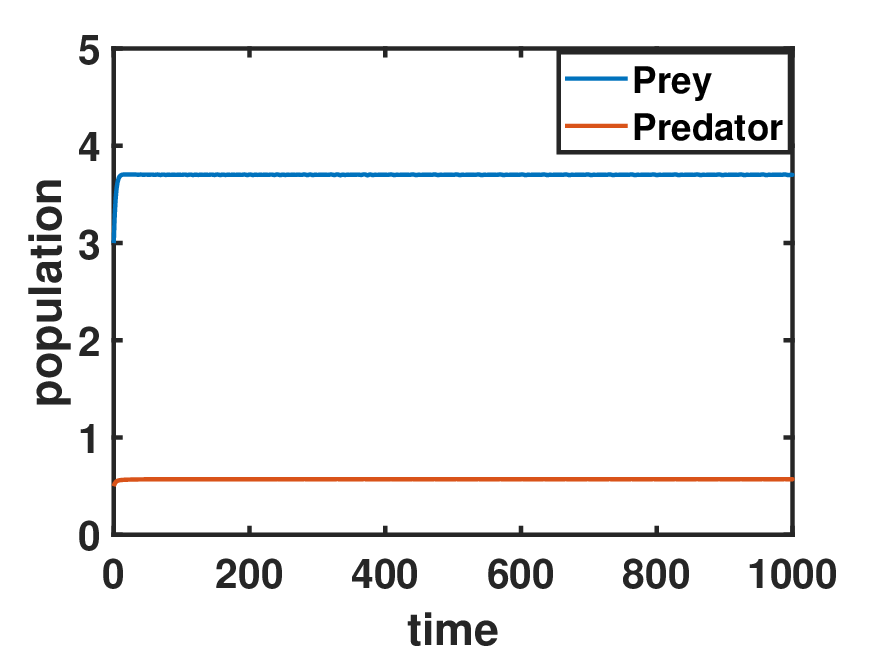}\label{TS_E*}}
     
	\caption{Time series diagram for the local stability of $E_2(k,0)$, $E_1(0,y_1)$, $E^*(x^*,y^*)$   of \eqref{eq:5} in figures \ref{TS_E2}, \ref{TS_E1}  and  \ref{TS_E*} respectively with the initial condition $(3,0.5)$ and $\phi(w)=0.61$, $\phi(w)=1.5$ and $\phi(w)=1$,  respectively and the others parameters are $r=0.5,k = 4,\alpha = 0.4,\xi=0.6,b = 0.1,\beta = 0.6,d = 0.1,c = 0.2.$}
	\end{figure}
    
 \begin{figure}[H]
	\centering
	\includegraphics[width=4.2 in, height=1.5 in]{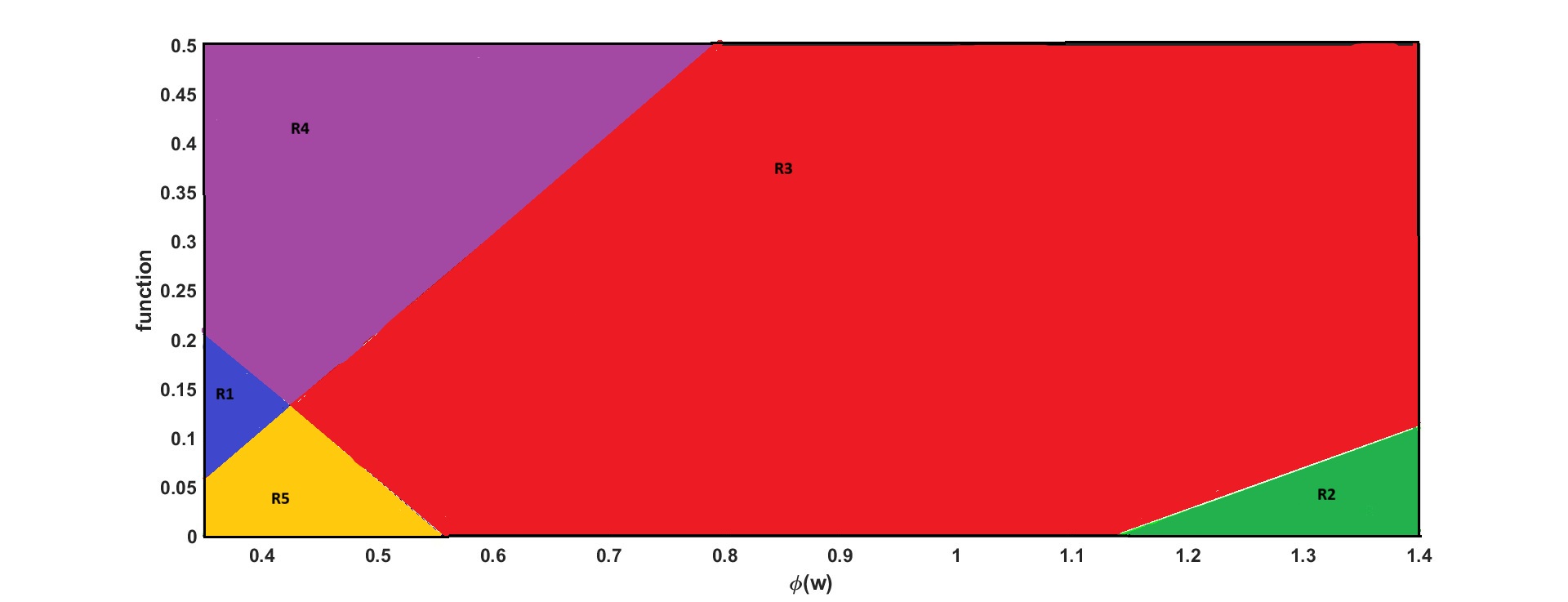}

	\caption{The figure shows the local stability of equilibriums of the model \eqref{eq:5} with respect to the wind speed function $\phi(w)$. The other parameters are $r = 1,\hspace{2mm}k = 5,\hspace{2mm}\alpha = 0.1,\hspace{2mm}\xi = 2,\hspace{2mm}b = 0.2,\hspace{2mm}\beta = 0.7,\hspace{2mm}d = 0.1,\hspace{2mm}c = 0.1.$ The region $R_1$(blue) is the region in which both the prey-free equilibrium $E_1 $ and the interior equilibrium $E^*$ are stable. In the region $R_2$(green), only the predator-free equilibrium $E_2 $ is stable. In the region $R_3$(red), both predator-free equilibrium $E_2$ and interior equilibrium $E^*$ are stable. The region $R_4$(magenta) is the region in which all the equilibriums are stable. In the region $R_5$(yellow), only interior equilibrium $E^*$ is stable.} 
    \label{R_P_1}
	
\end{figure}

\begin{table}[H]
\caption{Existence and stability analysis of equilibria of the system \eqref{eq:5}}
\centering
\begin{tabular}{|c|p{2.5cm}|p{3cm}|p{9.5cm}|p{0.8cm}|}
\hline
S.No. & Equilibrium Points & Feasibility Condition &  Status of Stability & proof \\ 
\hline
(i) & $ E_0(0,0)$ & Always (\ref{Th_Existence}) & Saddle point & (\ref{lemma:4.3}) \\ 
\hline
(ii) & $E_2(k,0)$ & Always (\ref{Th_Existence}) & \begin{tabular}{|p{9.1cm}|}
stable node if
    $\phi(w)<\dfrac{d(1+\alpha \xi + bk + k^2)}{\beta(k+\xi)}$ in Fig. \ref{TS_E2} \\ \hline  Saddle point if
    $\phi(w)>\dfrac{d(1+\alpha \xi + bk + k^2)}{\beta(k+\xi)}$\\ \hline  Neutral if
    $\phi(w)=\dfrac{d(1+\alpha \xi + bk + k^2)}{\beta(k+\xi)}$\\ 
\end{tabular} &  ( \ref{lemma:4.4}) \\ 
\hline 
(iii) & $E_1(0,y_1)$,$y_1=\dfrac{\beta\xi}{c\xi(1+\alpha\xi)}-\dfrac{d}{c\xi\phi(w)}$ &  $d-\dfrac{\beta\xi\phi(w)}{1+\alpha\xi} < 0$ (\ref{Th_Existence}) & \begin{tabular}{|p{9.1cm}|}
stable node if
    $\phi(w)>\dfrac{\bigg(d+r c\xi(1+\alpha\xi)\bigg)(1+\alpha\xi)}{\beta\xi}$ in Fig. \ref{TS_E1}\\ \hline Saddle point if
    $\phi(w)<\dfrac{\bigg(d+rc\xi(1+\alpha\xi)\bigg)(1+\alpha\xi)}{\beta\xi} $\\ \hline
   Neutral if
    $\phi(w)=\dfrac{\bigg(d+r c\xi(1+\alpha\xi)\bigg)(1+\alpha\xi)}{\beta\xi}$\\
\end{tabular} & ( \ref{lemma:4.5})\\
\hline
(iv)&  $E^*(x^*,y^*)$ & Any one of $A,B,C,D,E$ is negative (\ref{Th_Existence}) & \begin{tabular}{|p{9.1cm}|}
    asymptotically stable if $tr(J^*)<0$ and $det(J^*)>0$ in Fig. (\ref{TS_E*}) \\ \\ \\
    \hline 
    non-hyperbolic if $det(J^*)=0$\\ \\ \\
    \end{tabular} & ( \ref{lemma:4.6})\\
\hline
\end{tabular}
\label{tab:1}
\end{table}

 \begin{figure}[h]
	\centering
	\subfloat[]{\includegraphics[width=3.2 in, height=2 in]{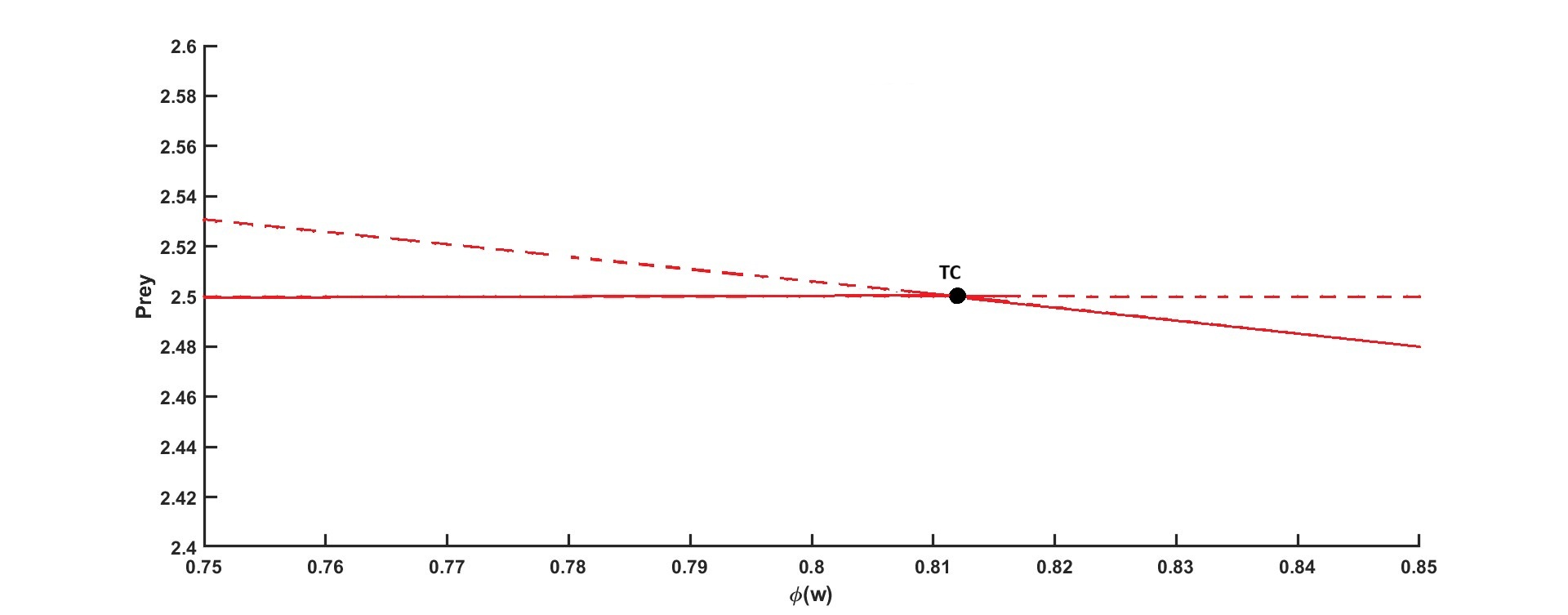}\label{TC_2}} 
    \subfloat[]{\includegraphics[width=3.2 in, height=2 in]{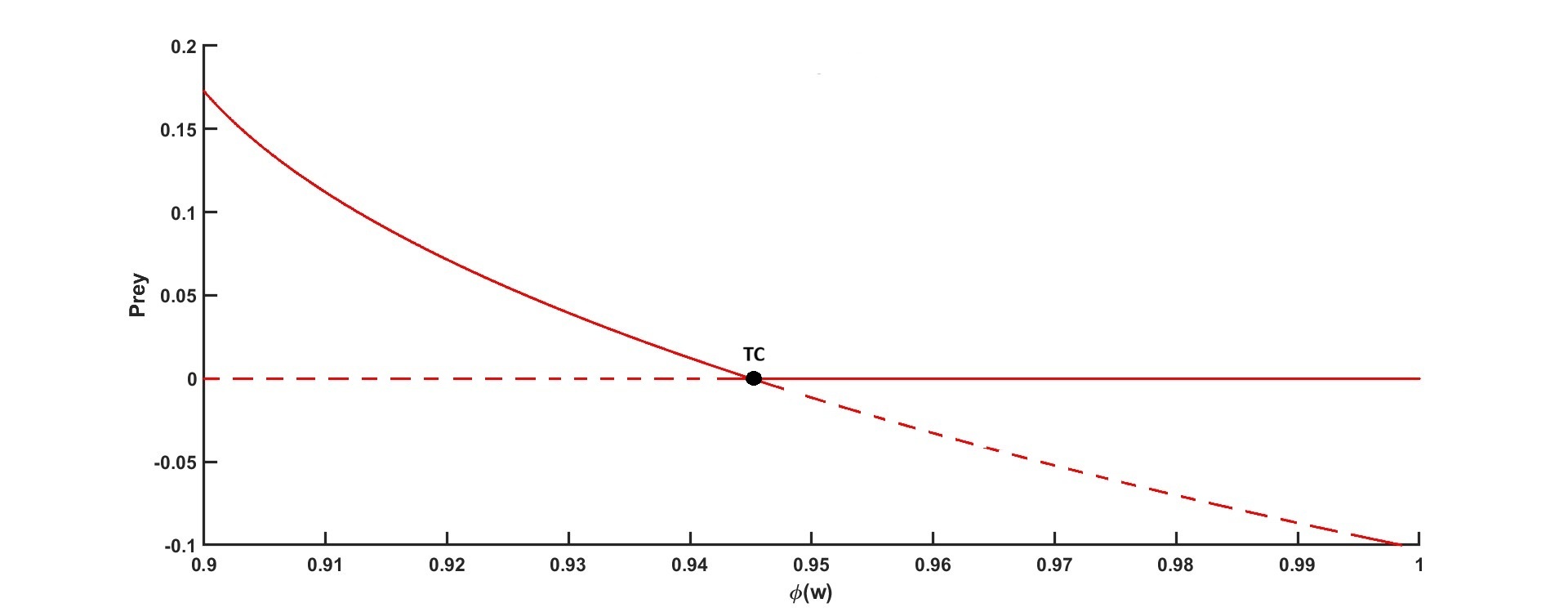}\label{TC_1}} 
	\caption{$r=0.5, \hspace{2mm}k=2.5,\hspace{2mm} b=0.5,\hspace{2mm}d=0.1,\hspace{2mm}\xi=1.5,\hspace{2mm}\beta=0.285,\hspace{2mm}\alpha=0.5,\hspace{2mm}c=0.1.$ \\
     Dotted lines indicates unstable and solid lines denote stable equilibriums. 
     Transcritical between axial $E_2$ and interior equilibrium at $\phi(w)=0.812$ in \ref{TC_2} and transcritical between axial $E_1$ and interior equilibrium at $\phi(w)=0.95$ in \ref{TC_1}. TC: Transcritical.}
\end{figure}
\begin{figure}[H]
	\centering
	\subfloat[]{\includegraphics[width=3 in, height=2 in]{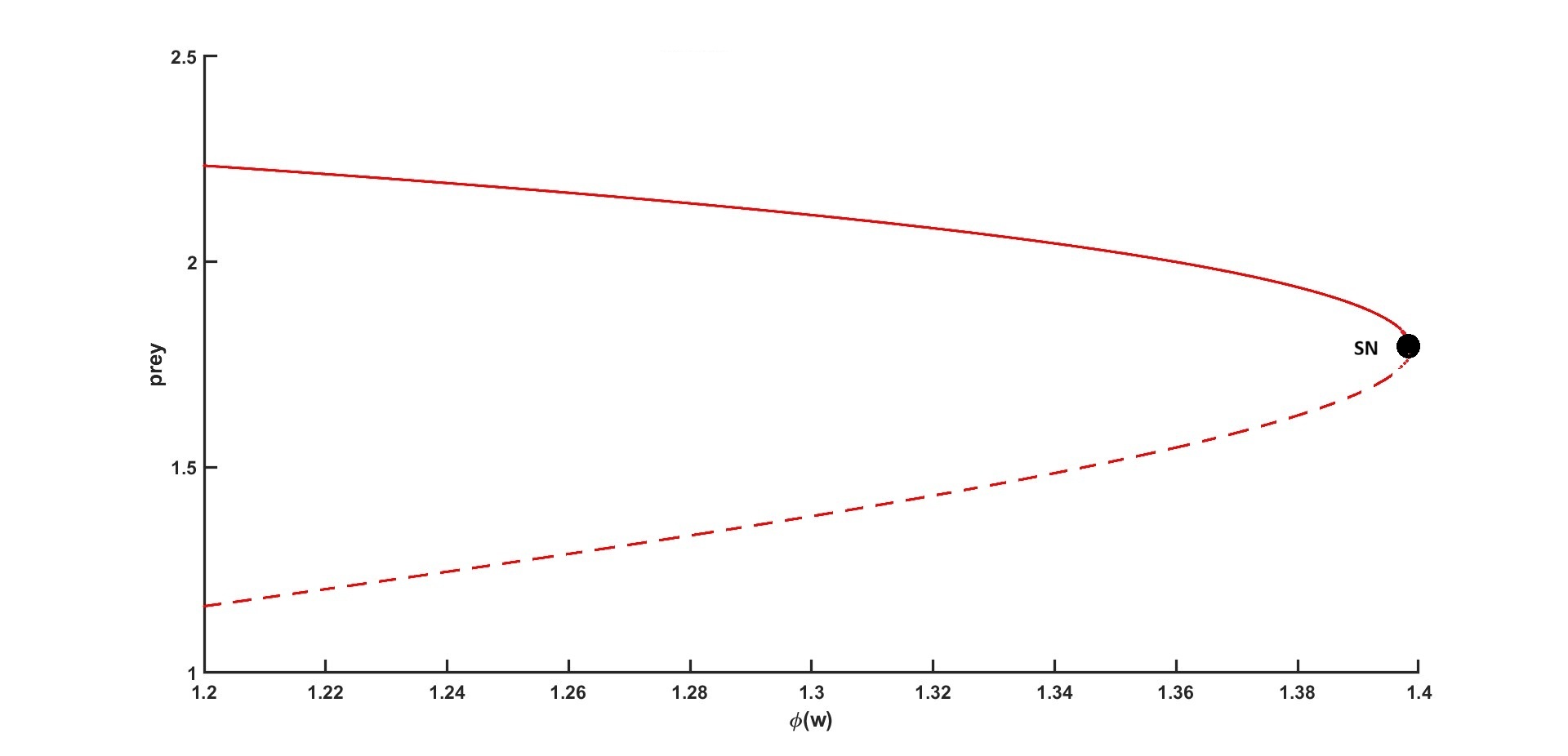}\label{SN_1}}
    \subfloat[]{\includegraphics[width=3 in, height=2 in]{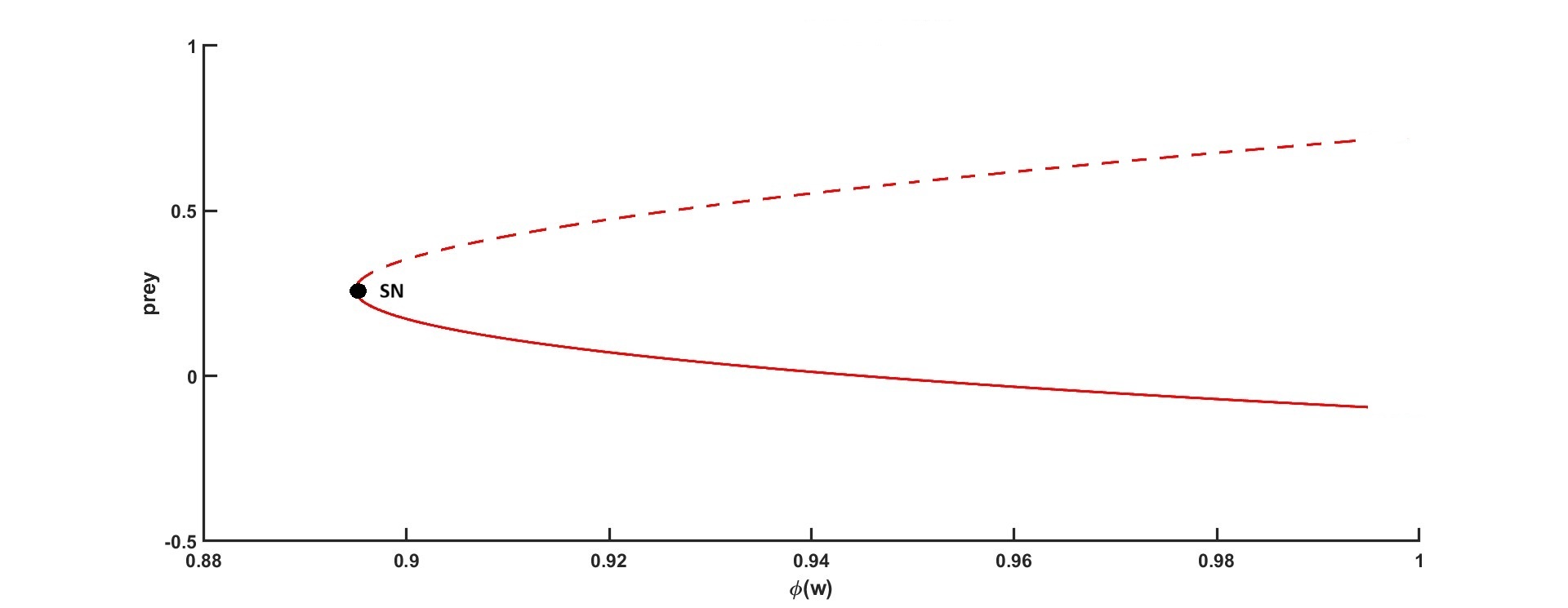}\label{SN_2}} 
	\caption{$r=0.5,\hspace{2mm} k=2.5,\hspace{2mm} b=0.5,\hspace{2mm}d=0.1,\hspace{2mm}\xi=1.5,\hspace{2mm}\beta=0.285,\hspace{2mm}\alpha=0.5,\hspace{2mm}c=0.1.$ \\
    Solid lines denote stable and dotted lines denotes unstable equilibrium.In figure \ref{SN_1} and \ref{SN_2}, at $\phi(w)=1.3985$ and at $\phi(w)=0.895$ respectively, saddle-node bifurcation occurs. SN: Saddle-node.}
\end{figure}

\begin{figure}[H]
    \centering
    \includegraphics[width=0.8\linewidth]{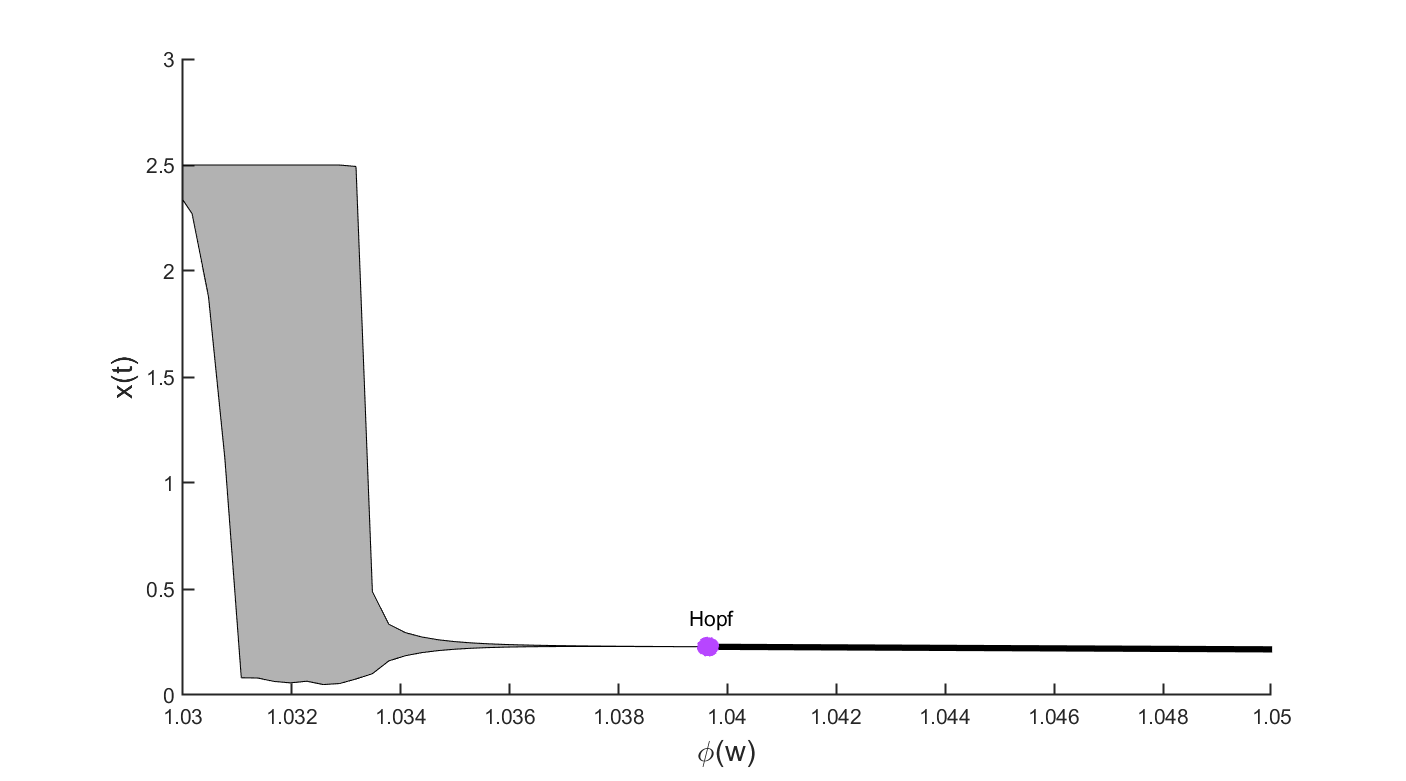}
    \caption{The figure shows the bifurcation diagram of $\phi(w)$ vs prey population. The grey shaded region shows the unstable limit cycle, and the black line signifies the stable equilibrium and this change is due to the Hopf bifurcation. The parameters used are $r=0.5,\hspace{2mm} k=2.5,\hspace{2mm} b=0.5,\hspace{2mm}d=0.1,\hspace{2mm}\xi=1.5,\hspace{2mm}\beta=0.285,\hspace{2mm}\alpha=0.5,\hspace{2mm}c=0.1.$}
    \label{Hopf_bifurcation}
\end{figure}

\begin{figure}[H]
    \centering

    \begin{subfigure}{0.45\textwidth}
        \centering
        \includegraphics[width=\linewidth]{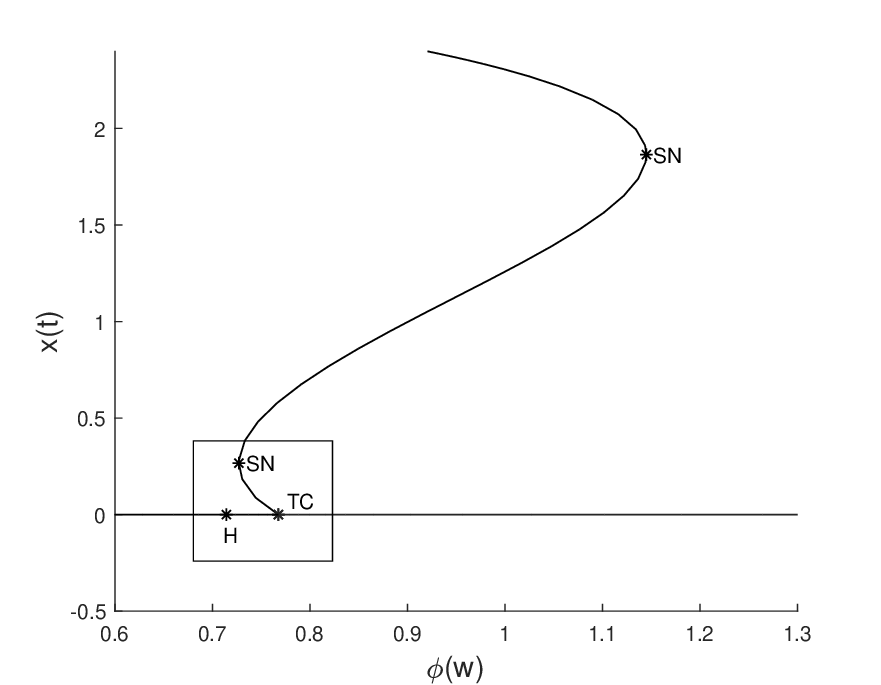}
        \caption{}
        \label{bt_1}
    \end{subfigure}
    \hfill
    \begin{subfigure}{0.45\textwidth}
        \centering
        \includegraphics[width=\linewidth]{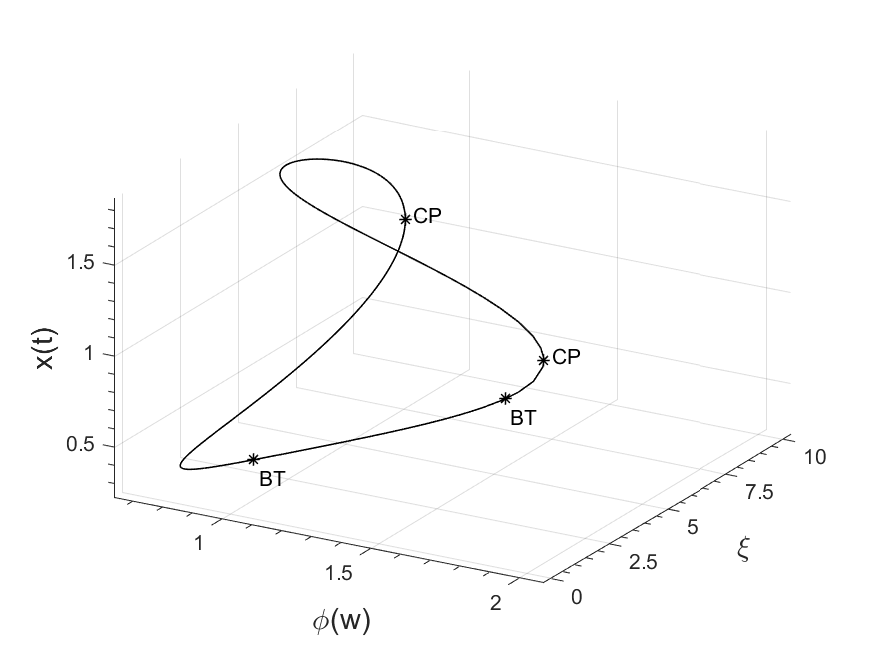}
        \caption{}
        \label{bt_2}
    \end{subfigure}

    \vspace{0.5cm}

    \begin{subfigure}{0.45\textwidth}
        \centering
        \includegraphics[width=\linewidth]{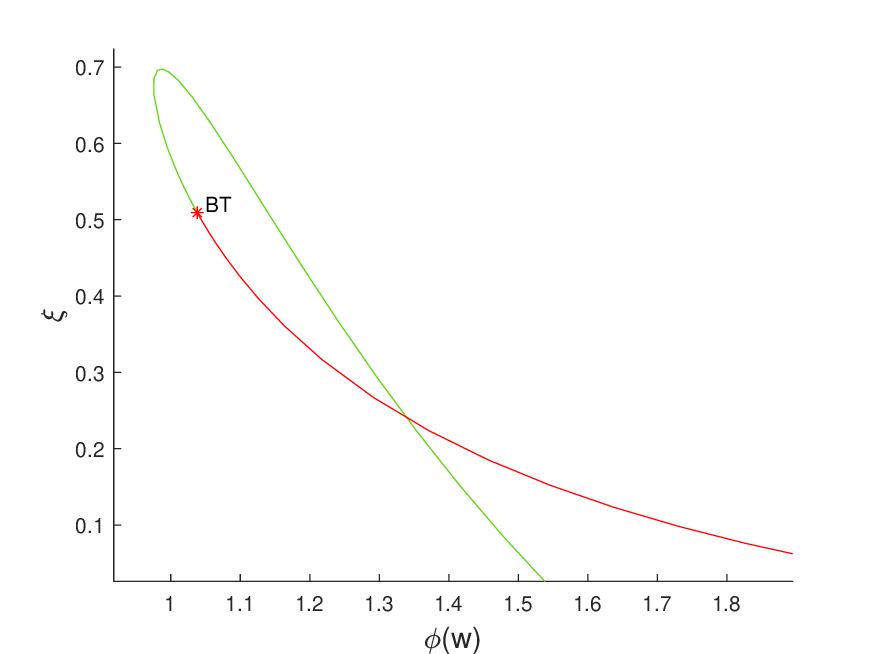}
        \caption{}
        \label{bt_hopf_curve}
    \end{subfigure}
    \hfill
    \begin{subfigure}{0.45\textwidth}
        \centering
        \includegraphics[width=\linewidth]{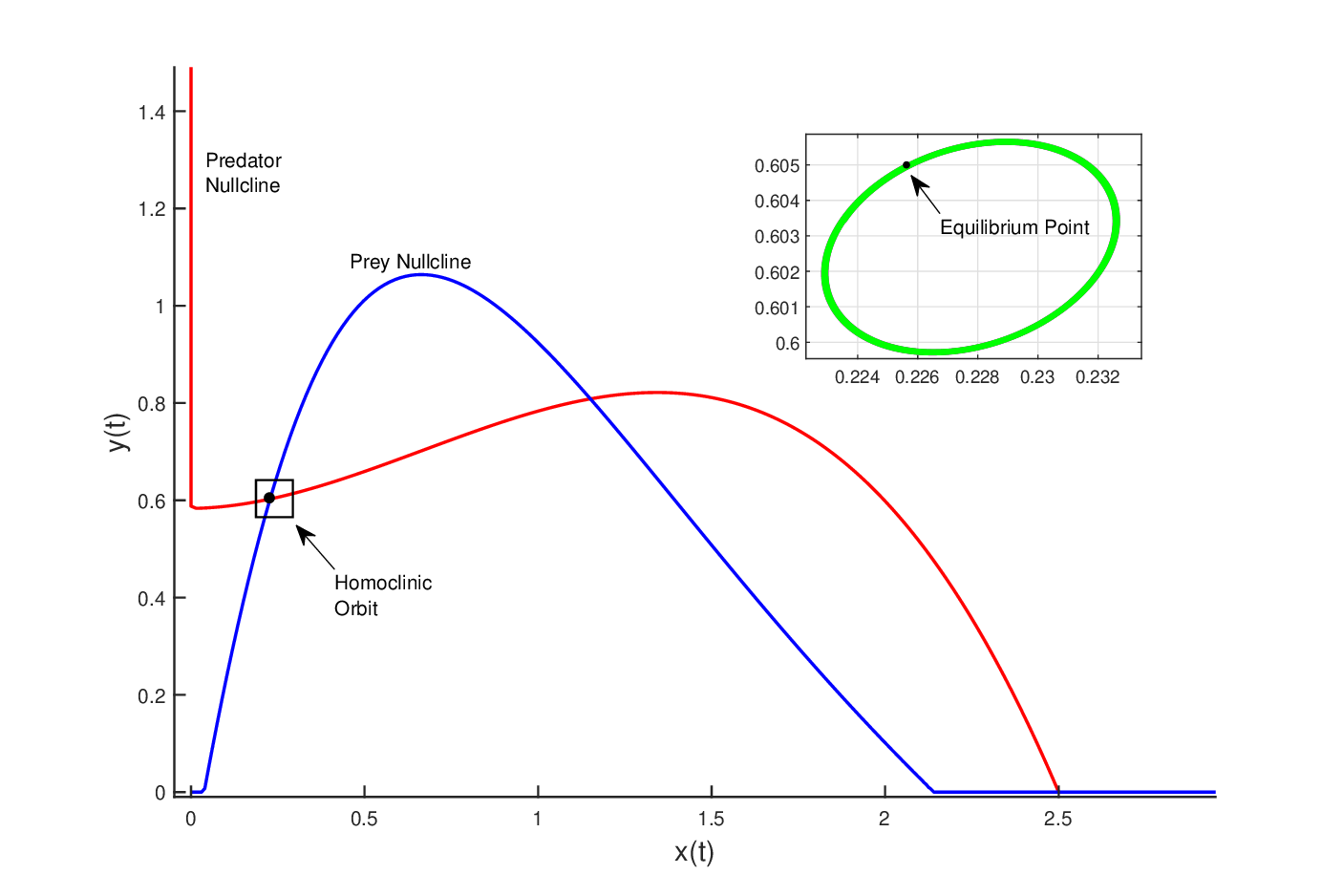}
        \caption{}
        \label{homoclinic}
    \end{subfigure}

    \caption{The figure shows the existence of Bogdanov-Takens bifurcation and dynamics around the neighborhood. Figure \ref{bt_1} and \ref{bt_2} show the codimension 1 and 2 bifurcation diagrams, respectively. Figure \ref{bt_hopf_curve} shows the Hopf curve with two changing parameters $\xi$ and $\phi(w)$. Figure \ref{homoclinic} shows the prey nullcline (Blue) and predator nullcline (red), and the boxed region is the equilibrium where a homoclinic orbit exists. The zoomed-in image shows the existence of a homoclinic orbit around the same BT neighborhood. SN: Saddle-node, TC: Transcritical, BT: Bogdanov-Takens, CP: Cusp. The parameters used are $r=0.5,\hspace{2mm} k=2.5,\hspace{2mm} b=0.5,\hspace{2mm}d=0.1,\hspace{2mm}\beta=0.285,\hspace{2mm}\alpha=0.5,\hspace{2mm}c=0.1.$}
    \label{BT_dynamics}
\end{figure} 

\begin{figure}[H]
	\centering
	\subfloat[]{\includegraphics[width=3.2 in, height=2 in]{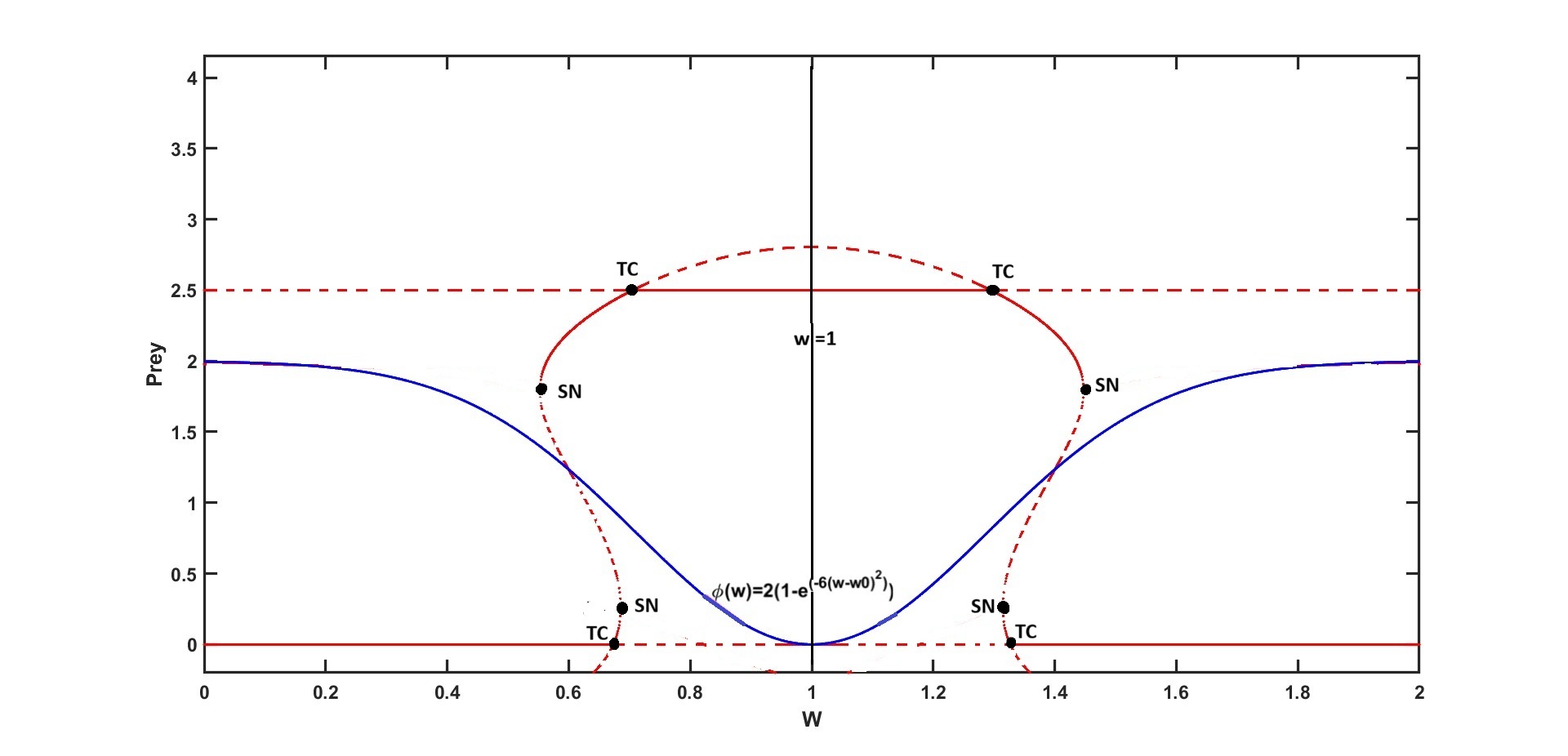} \label{2exp_1}}
    \subfloat[]{\includegraphics[width=3.2 in, height=2 in]{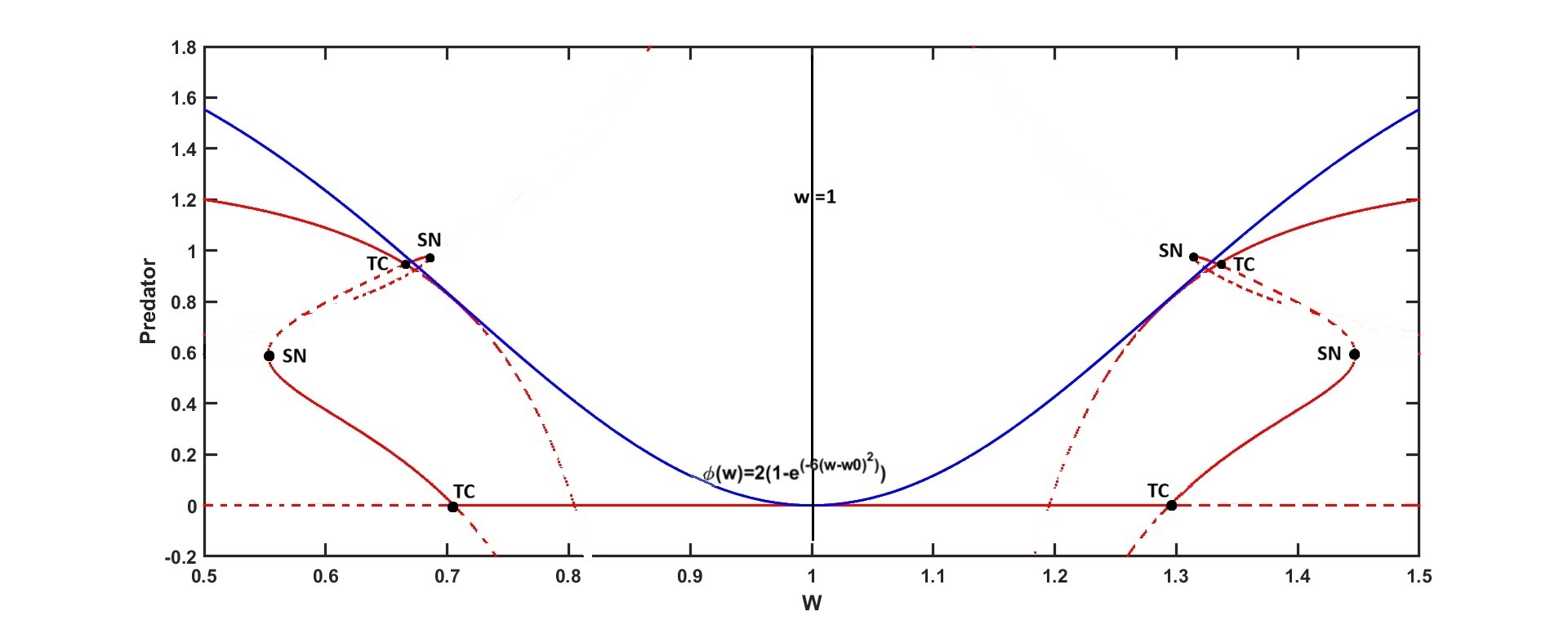}\label{2exp_2}}
    \\
    \subfloat[]{\includegraphics[width=3.2 in, height=2 in]{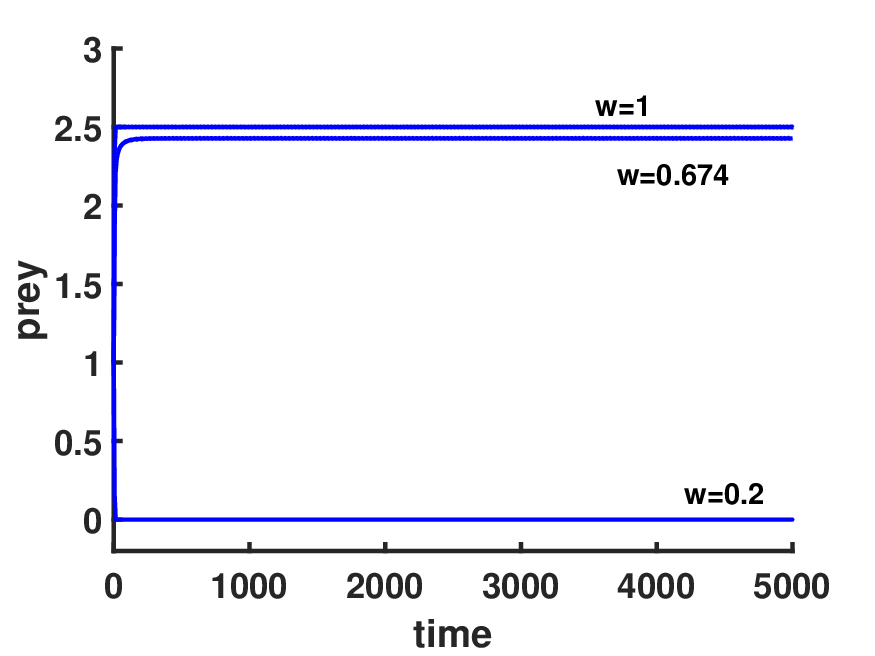}\label{2exp_3}}
    \subfloat[]{\includegraphics[width=3.2 in, height=2 in]{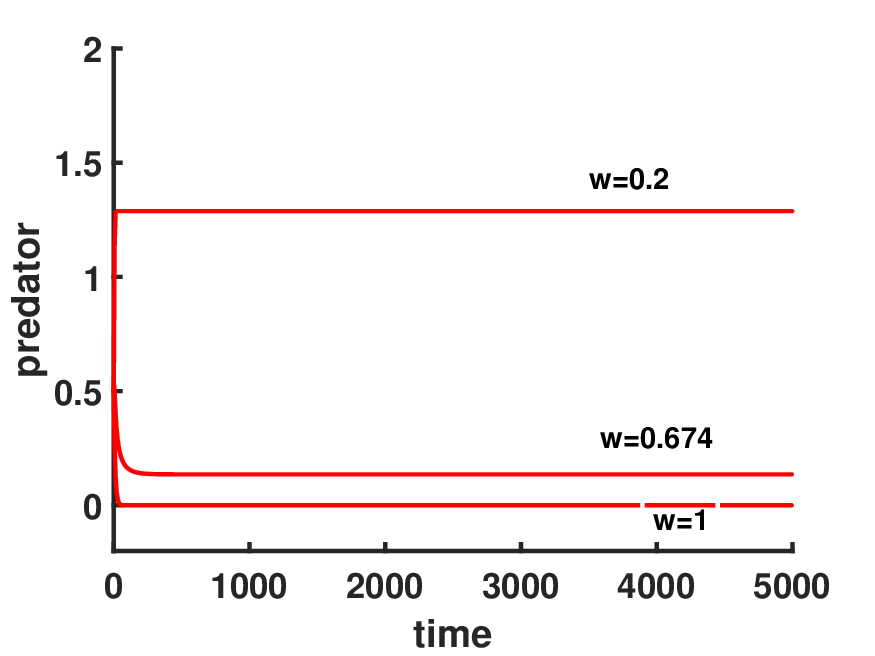}\label{2exp_4}}\\
    \subfloat[]{\includegraphics[width=3.2 in, height=2 in]{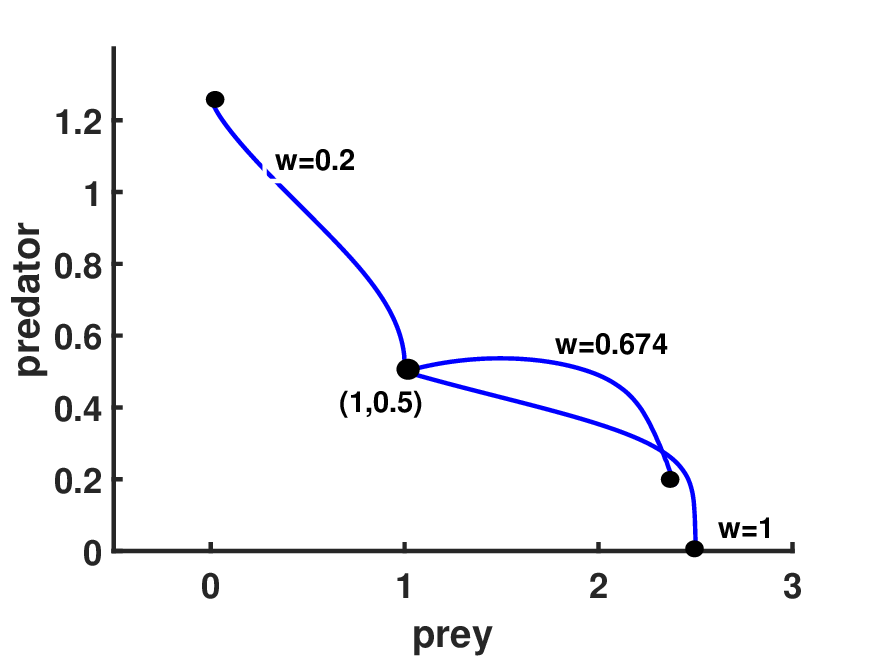}\label{2exp_5}}
    \subfloat[]{\includegraphics[width=3.2 in, height=2 in]{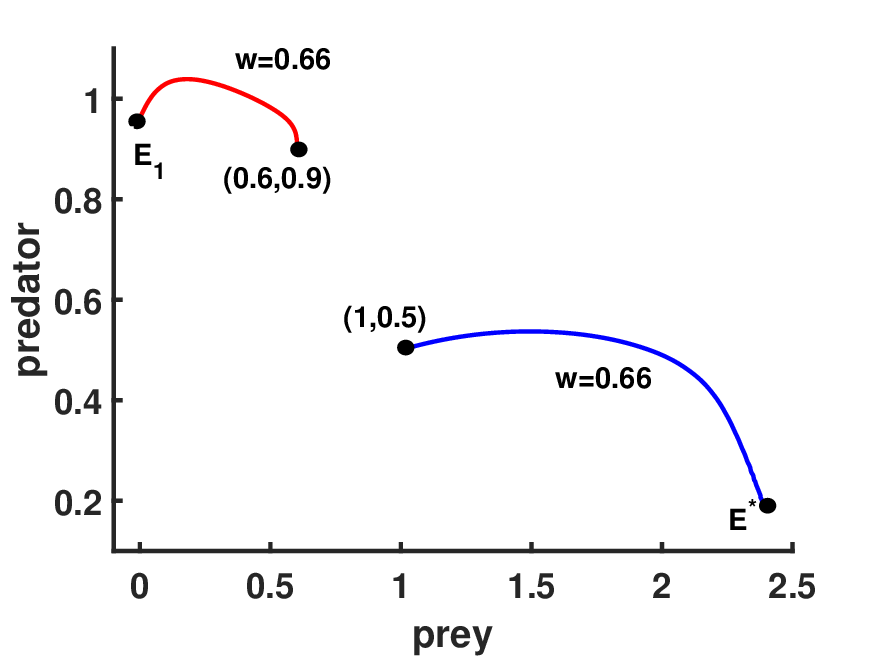}\label{2exp_6}}
   
   \caption{Here we choose a bounded U-shaped function  $\phi(w)=2\bigg(1-e^{(-6(w-w_0)^2)}\bigg)$,  where $w_0$ is the optimum wind speed. Figs. \ref{2exp_1} and \ref{2exp_2} represents the bifurcation figures of wind speed $w$ vs prey and predator  population, respectively. In Figs. \ref{2exp_1} and \ref{2exp_2}, dotted red lines represents unstable equilibrium and solid red lines indicates stable equilibrium and solid blue line represents the function $\phi(w)$. Figs. \ref{2exp_3} and \ref{2exp_4} shows the time series for prey and predator population with different $w$ and fig. \ref{2exp_5} is the phase portrait of the respective population. The others parameters are r=0.5, k=2.5, b=0.5, d=0.1, $\xi=1.5,\hspace{2mm}$ $\beta=0.285,$ $\alpha=0.5$ c=0.1, $w_0$=1. SN: Saddle-Node, TC: Transcritical.
   } 
\end{figure}
\begin{figure}[H]
	\centering
	\subfloat[]{\includegraphics[width=3.2 in, height=2 in]{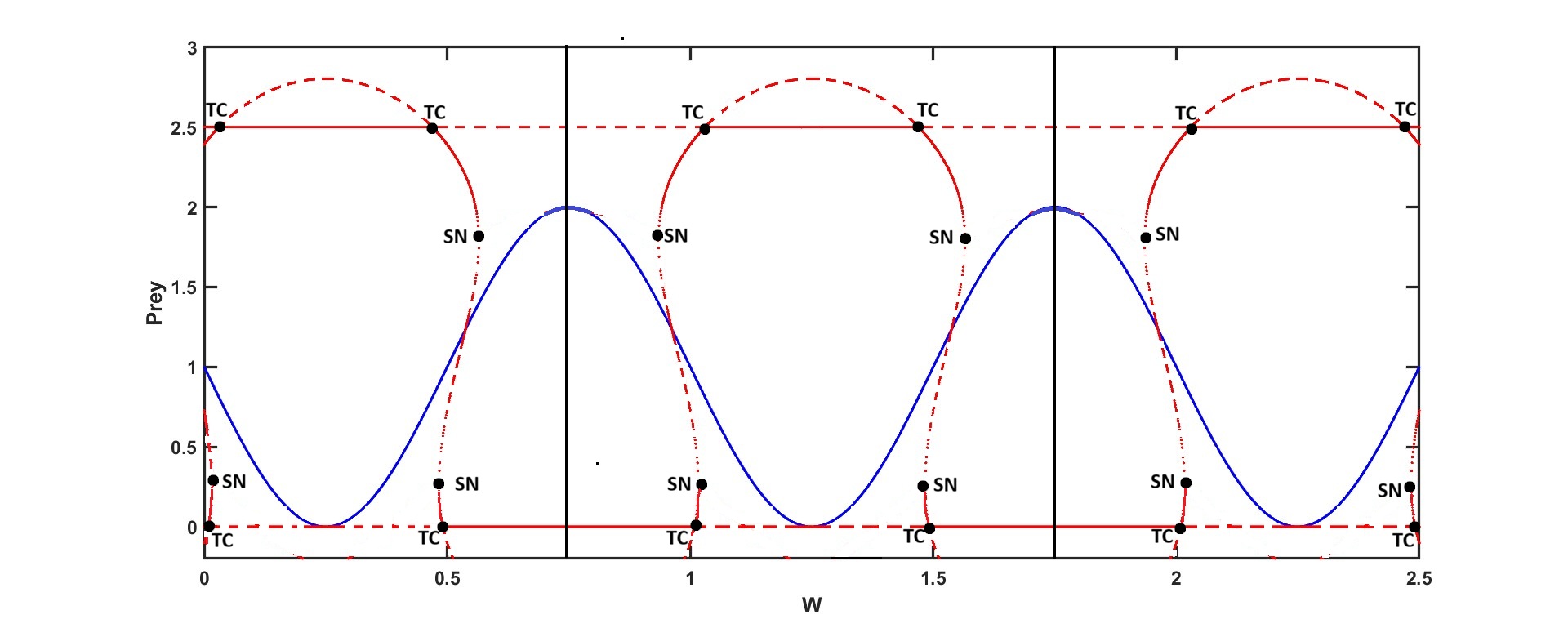} \label{sin_1}}
    \subfloat[]{\includegraphics[width=3.2 in, height=2 in]{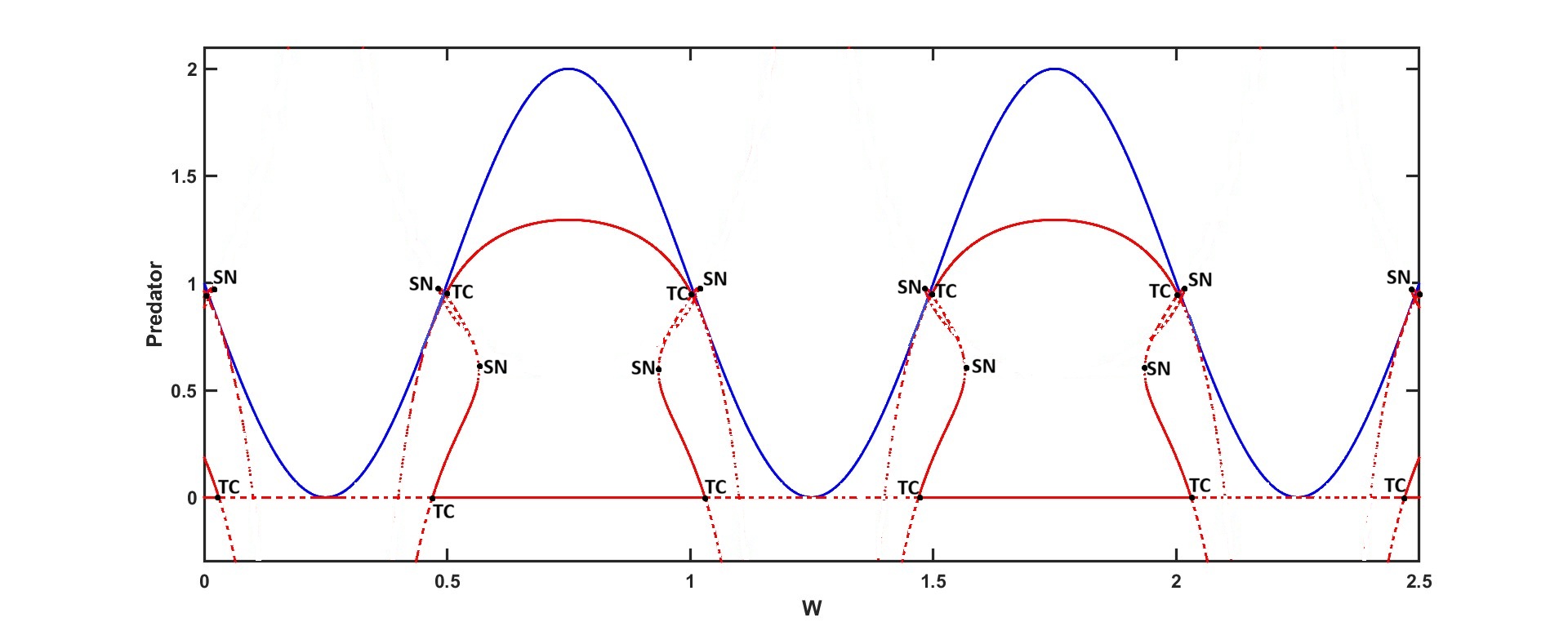}\label{sin_2}}
    \\
    \subfloat[]{\includegraphics[width=3.2 in, height=2 in]{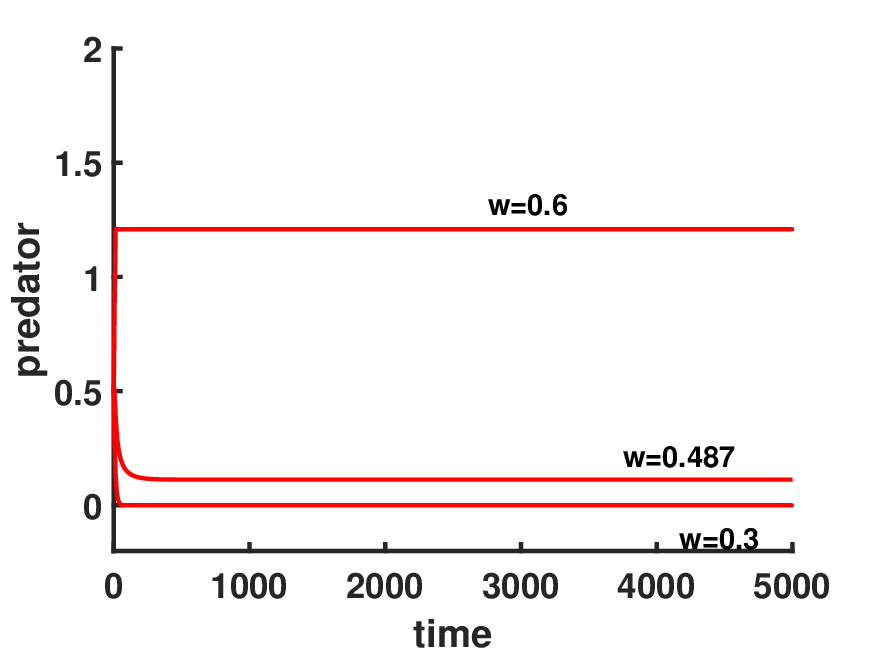}\label{sin_3}}
    \subfloat[]{\includegraphics[width=3.2 in, height=2 in]{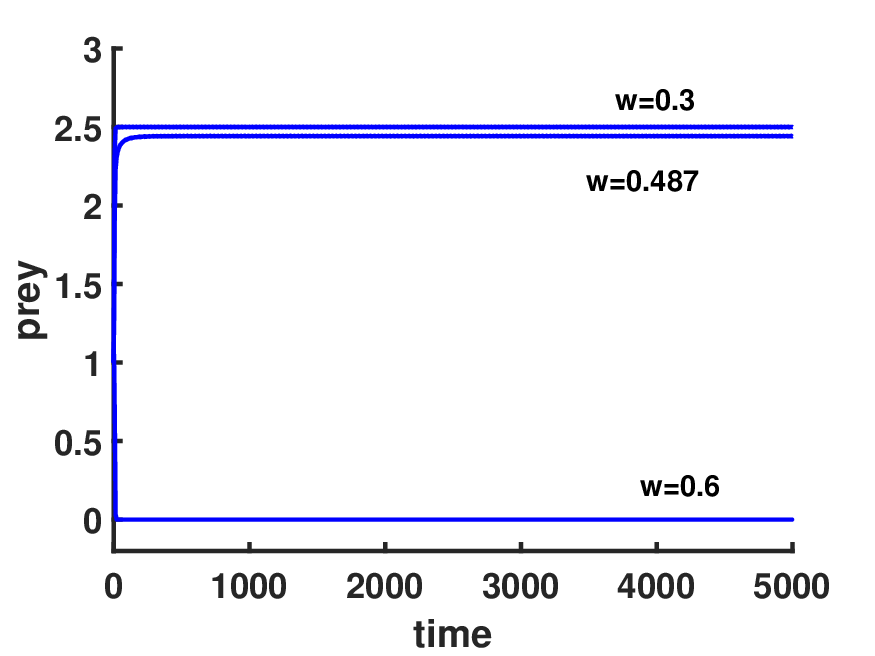}\label{sin_4}}\\
    \subfloat[]{\includegraphics[width=3.2 in, height=2 in]{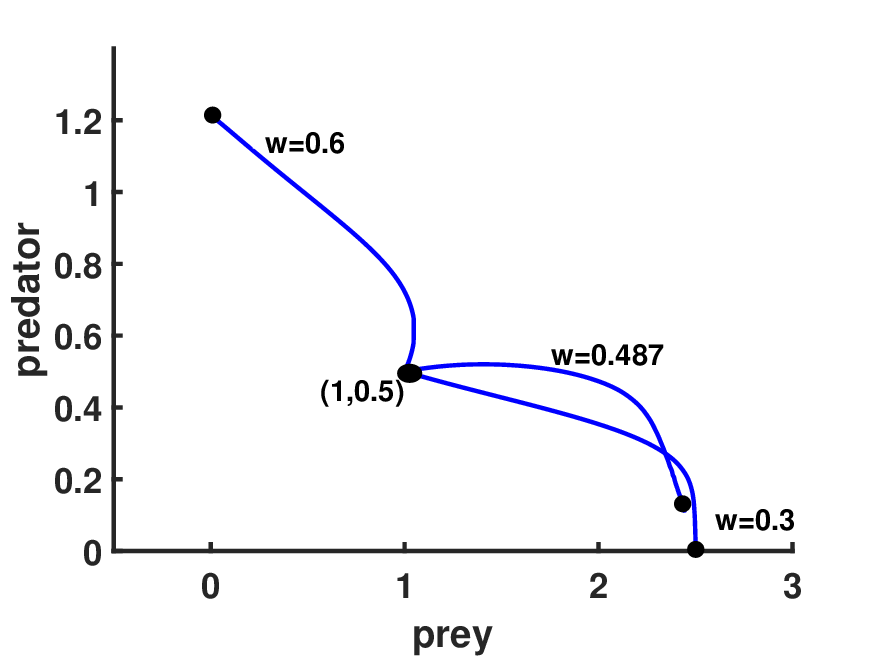}\label{sin_5}}
    \caption{Here we choose a periodic function $\phi(w)=1-sin(2\pi w)$. Figs. \ref{sin_1} and \ref{sin_2} shows the bifurcation figures of wind speed $w$ vs prey and predator  population, respectively. In Figs. \ref{sin_1} and \ref{sin_2}, dotted red lines represent unstable equilibrium, and solid red lines indicate stable equilibrium, and the solid blue line represents the function $\phi(w)$. Figs. \ref{sin_3} and \ref{sin_4} show the time series for prey and predator population with different $w$ and fig. \ref{sin_5} is the phase portrait of the respective population dynamics. The others parameters are r=0.5, k=2.5, b=0.5, d=0.1, $\xi=1.5,$ $\beta=0.285,$ $\alpha=0.5,$ c=0.1. SN: Saddle-Node, TC: Transcritical.
  } 
\end{figure}

 \begin{figure}[htb!]
	\centering
	\subfloat[]{\includegraphics[width=3.2 in, height=2 in]{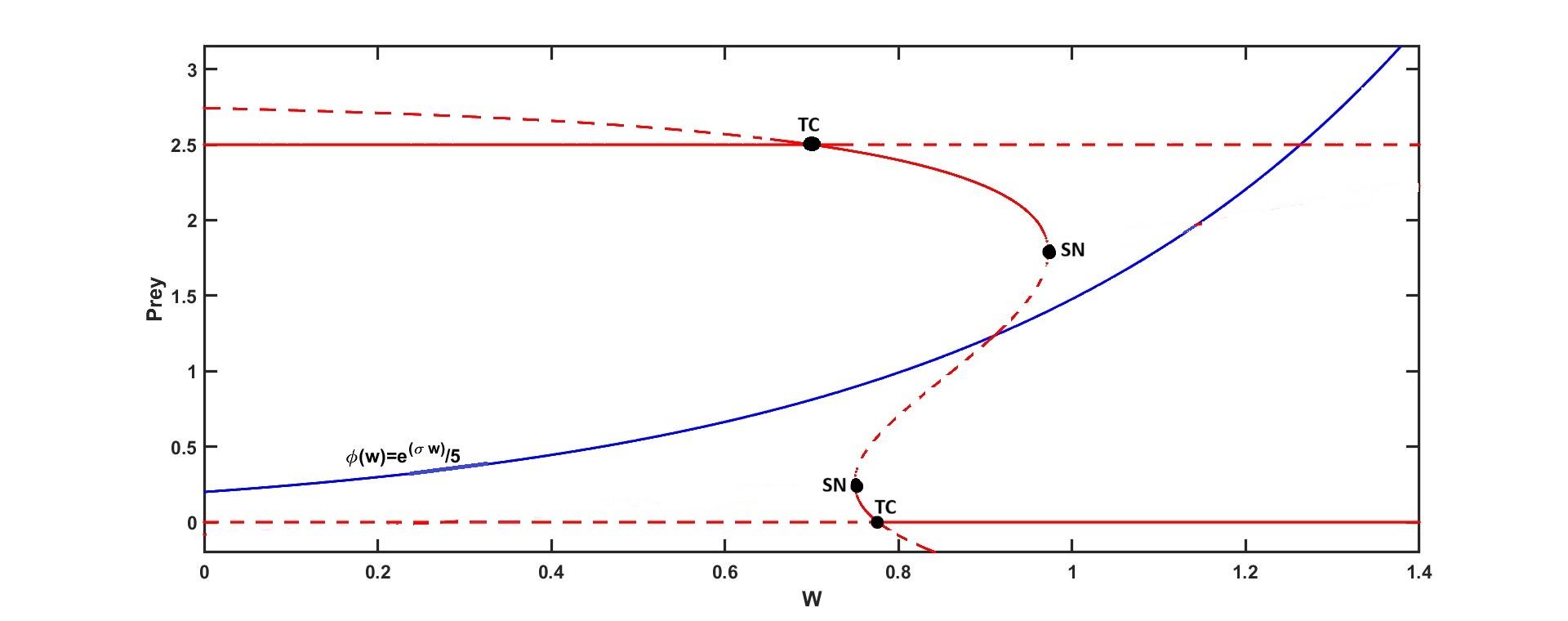} \label{exp^sigma_1}}
    \subfloat[]{\includegraphics[width=3.2 in, height=2 in]{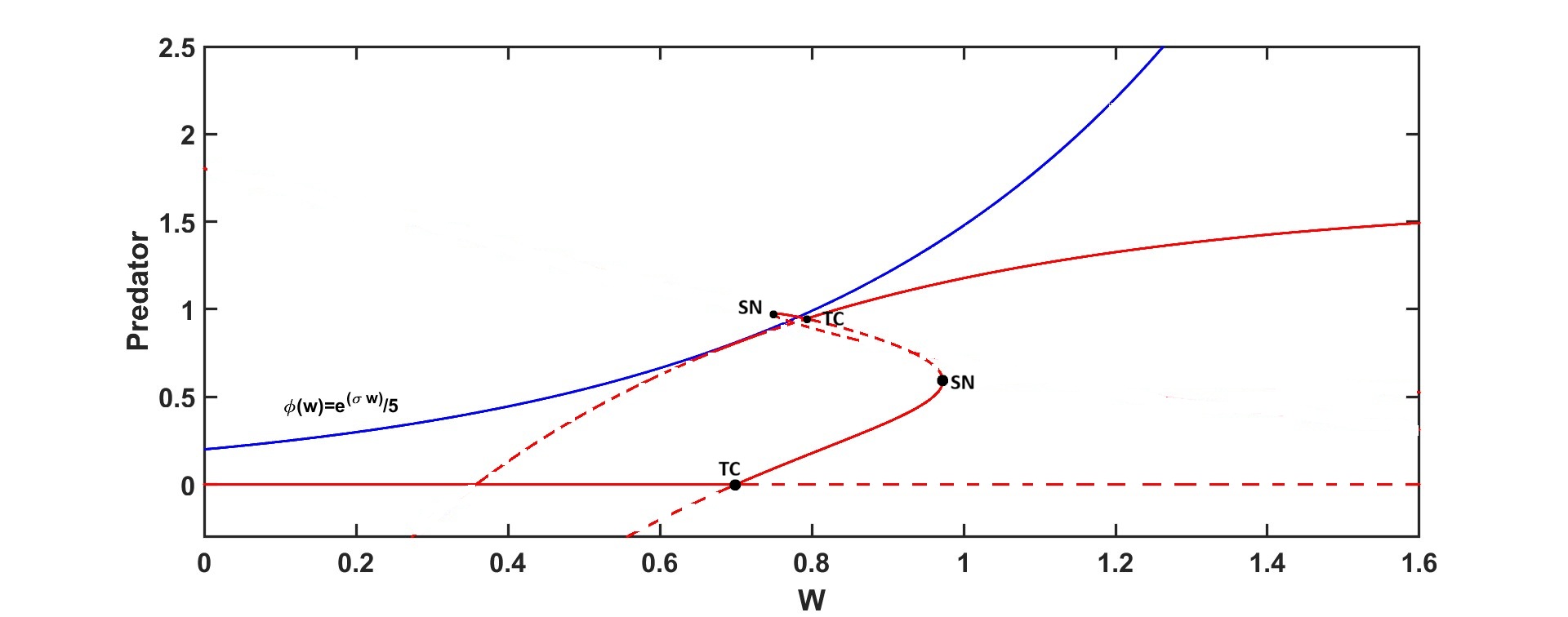}\label{exp^sigma_2}}
    \\
    \subfloat[]{\includegraphics[width=3.2 in, height=2 in]{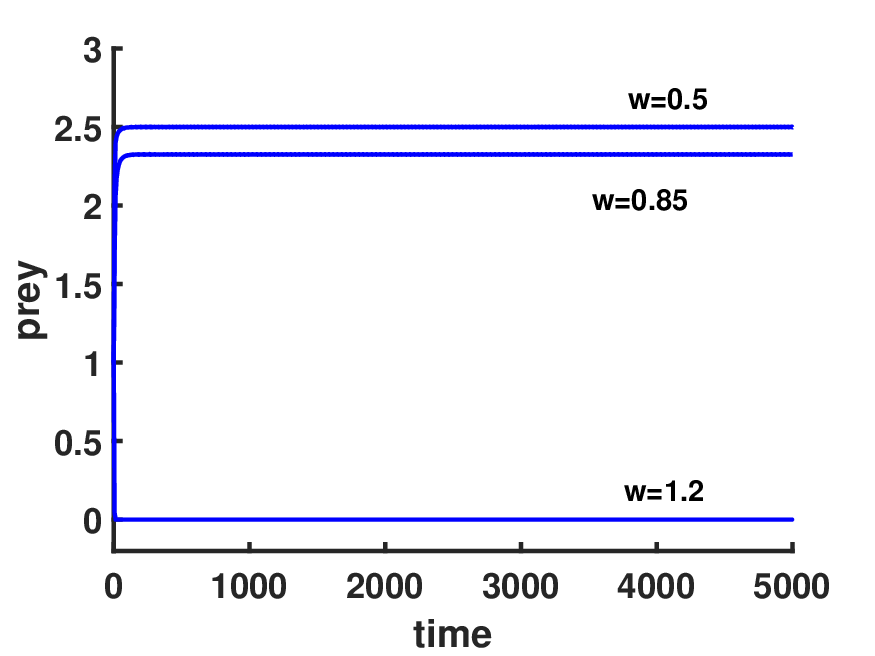}\label{exp^sigma_3}}
    \subfloat[]{\includegraphics[width=3.2 in, height=2 in]{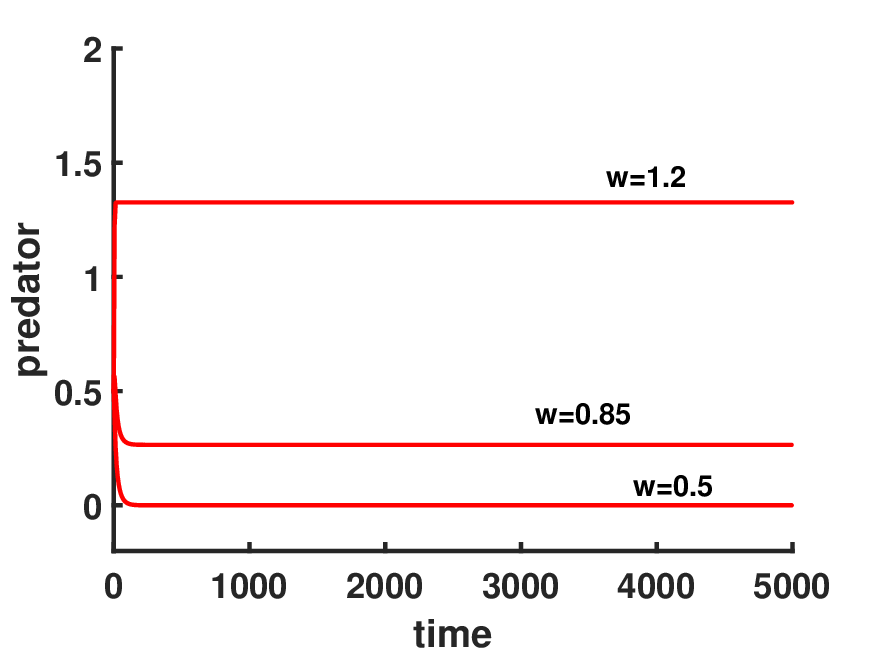}\label{exp^sigma_4}}\\
    \subfloat[]{\includegraphics[width=3.2 in, height=2 in]{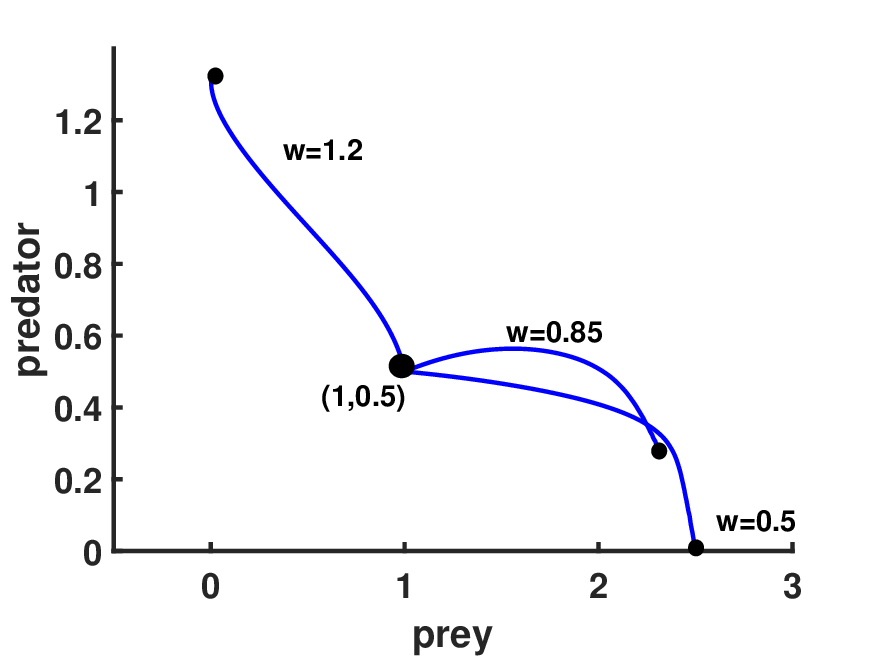}\label{exp^sigma_5}}
      \subfloat[]{\includegraphics[width=3.2 in, height=2 in]{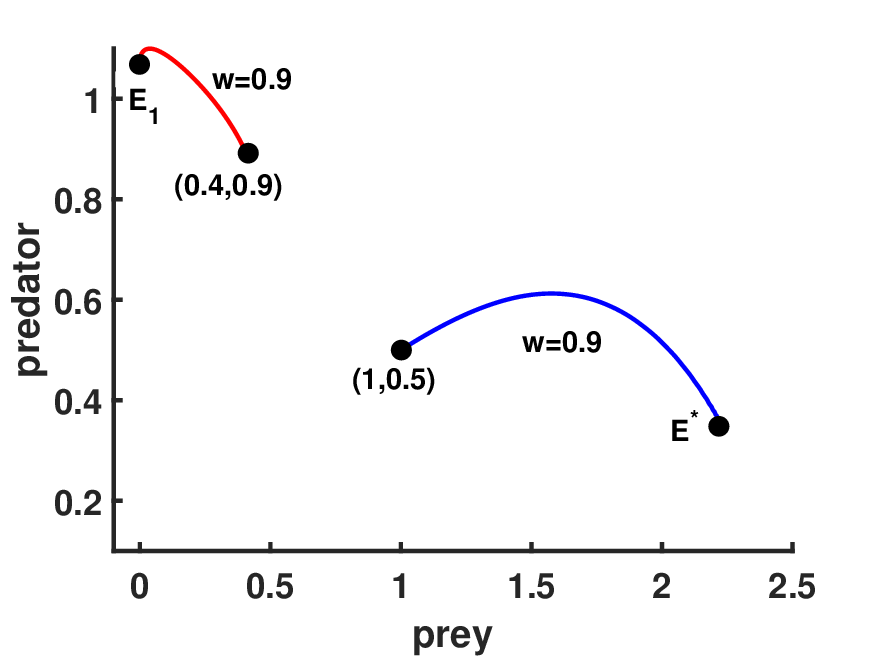}\label{exp^sigma_6}}
    
\caption{We choose an strictly increasing  function  $\phi(w)=\dfrac{e^{\sigma w}}5$. Figs. \ref{exp^sigma_1} and \ref{exp^sigma_2}  shows the bifurcation figures of wind speed $w$ vs prey and predator  population, respectively. In Figs. \ref{exp^sigma_1} and \ref{exp^sigma_2}, dotted red lines and solid red lines represents unstable and stable equilibrium, respectively and solid blue line represents the function $\phi(w)$. Figs. \ref{exp^sigma_3} and \ref{exp^sigma_4} shows the time series for prey and predator population with different $w$ and fig. \ref{exp^sigma_5} is the phase portrait of the respective population dynamics. The others parameters are r=0.5, k=2.5, b=0.5, d=0.1, $\xi=1.5,$ $\beta=0.285,$ $\alpha=0.5,$ c=0.1, $\sigma=2$. SN: Saddle-Node, TC: Transcritical. } 
\end{figure}

\begin{figure}[ht!]
	\centering
	\subfloat[]{\includegraphics[width=3.2 in, height=2 in]{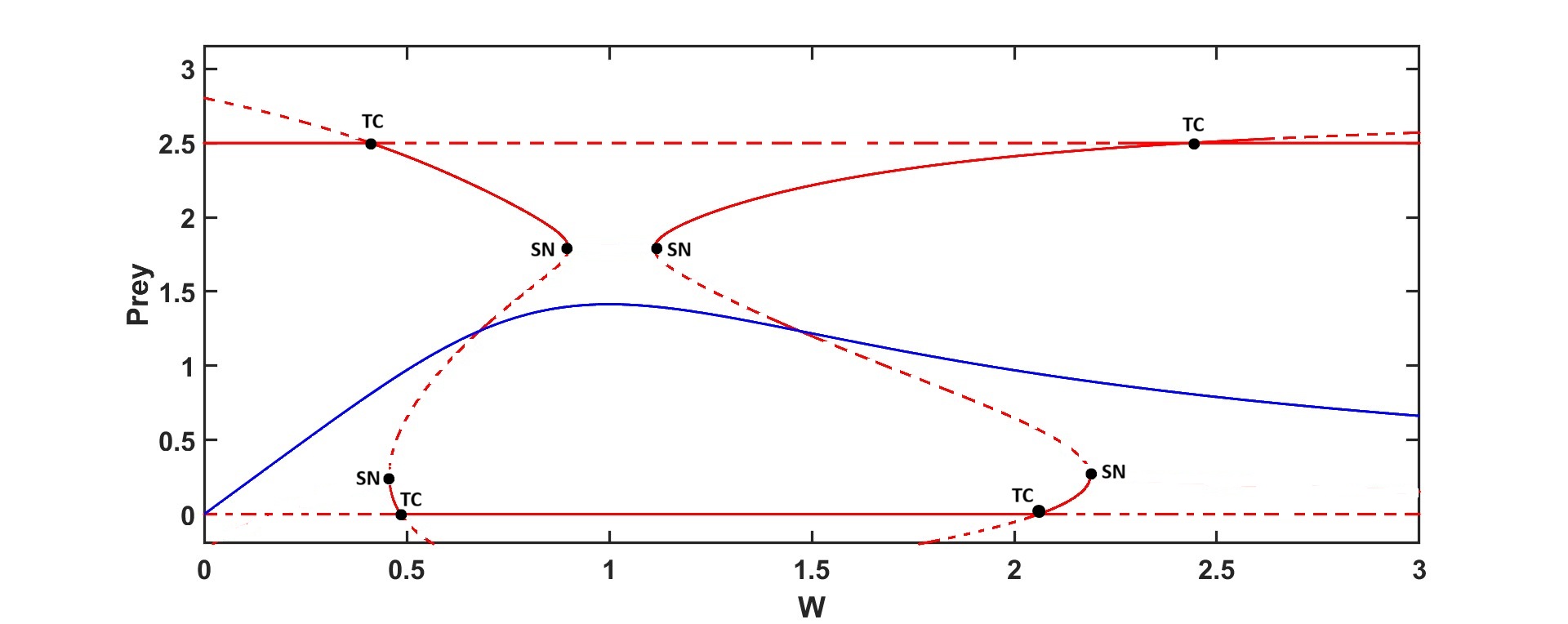} \label{square_1}}
    \subfloat[]{\includegraphics[width=3.2 in, height=2 in]{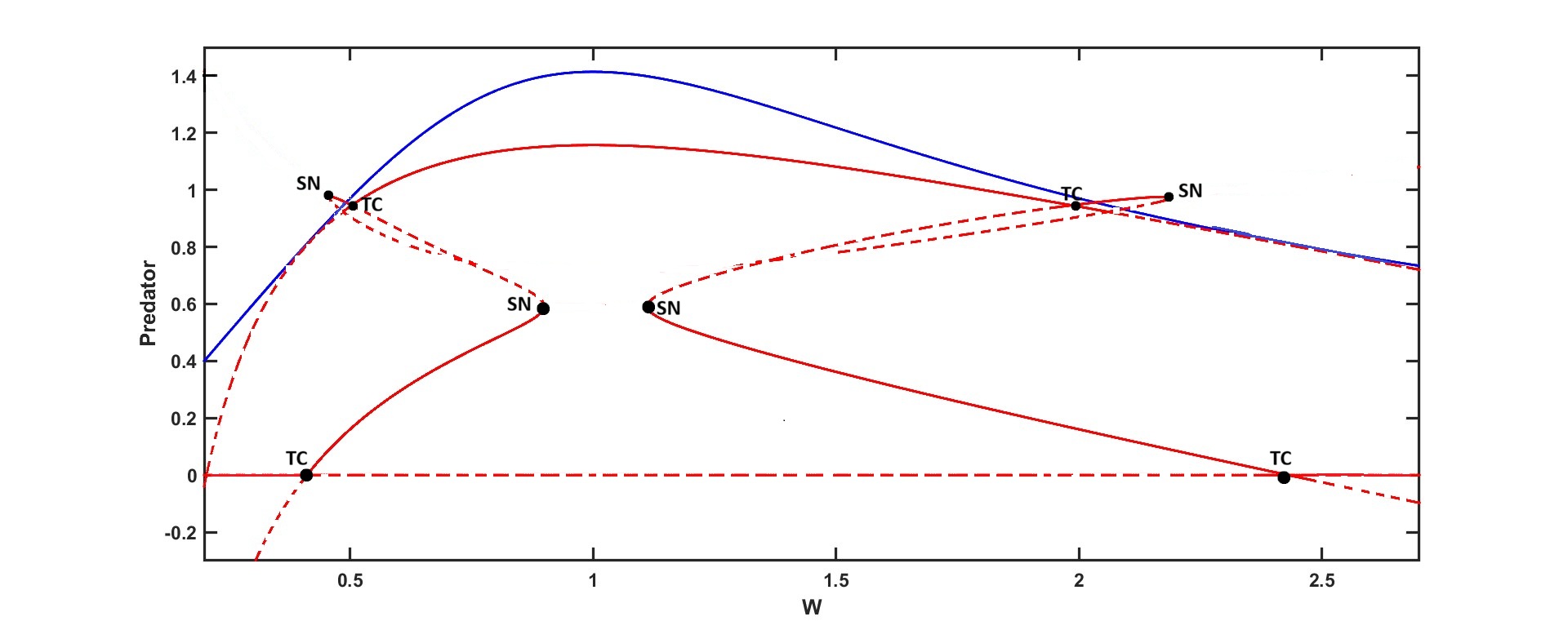}\label{square_2}}
    \\
    \subfloat[]{\includegraphics[width=3.2 in, height=2 in]{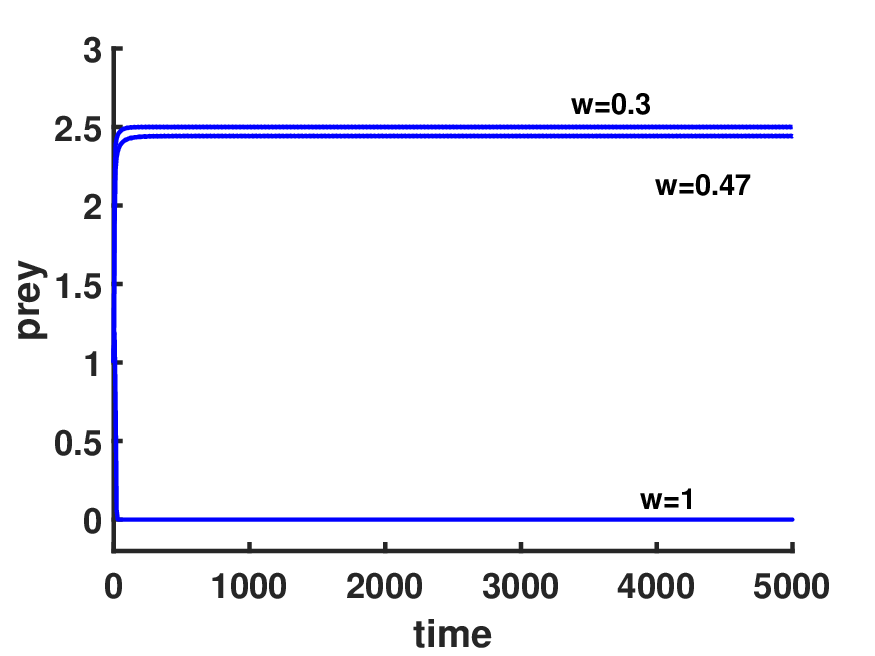}\label{square_3}}
    \subfloat[]{\includegraphics[width=3.2 in, height=2 in]{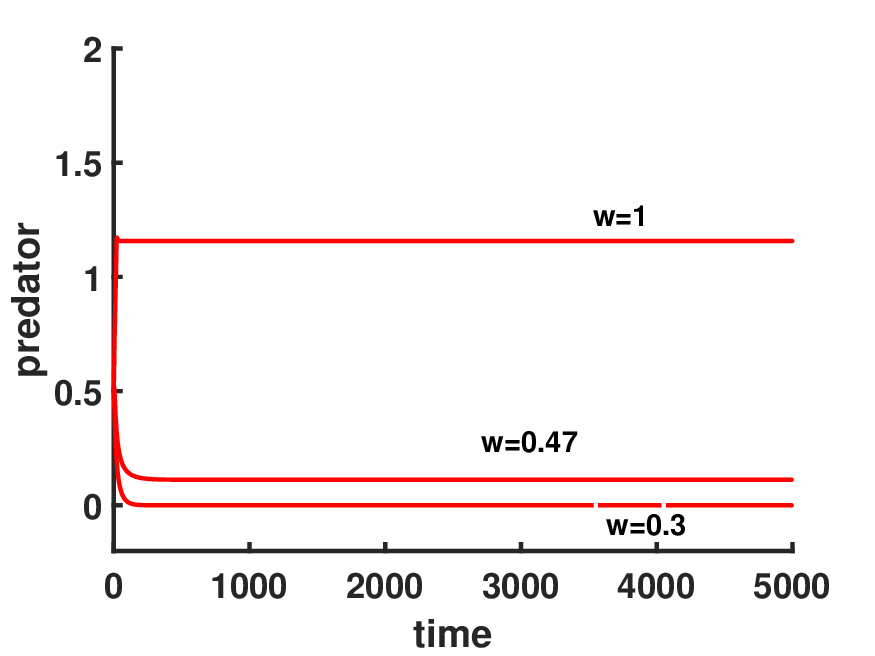}\label{square_4}}\\
    \subfloat[]{\includegraphics[width=3.2 in, height=2 in]{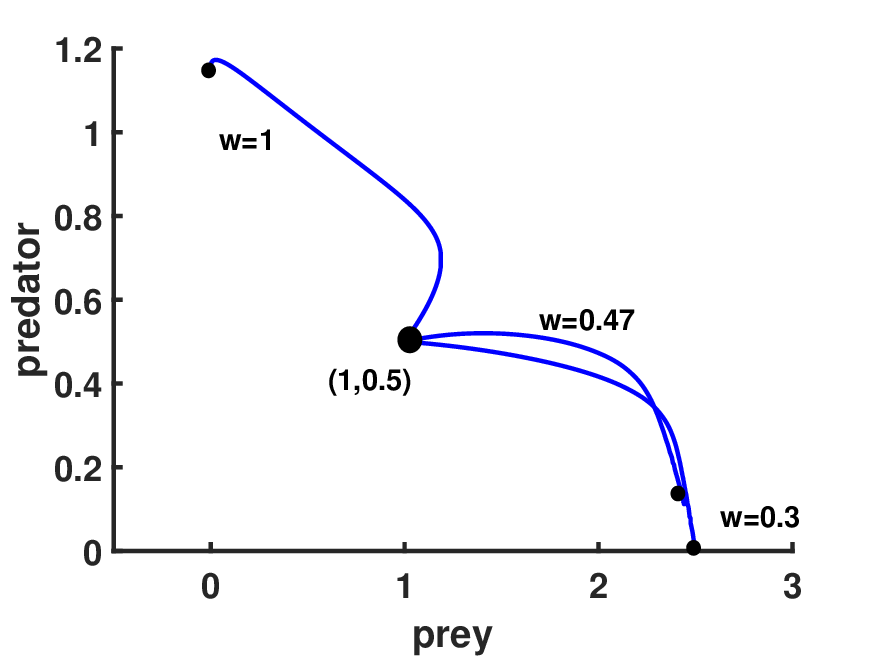}\label{square_5}}
    \subfloat[]{\includegraphics[width=3.2 in, height=2 in]{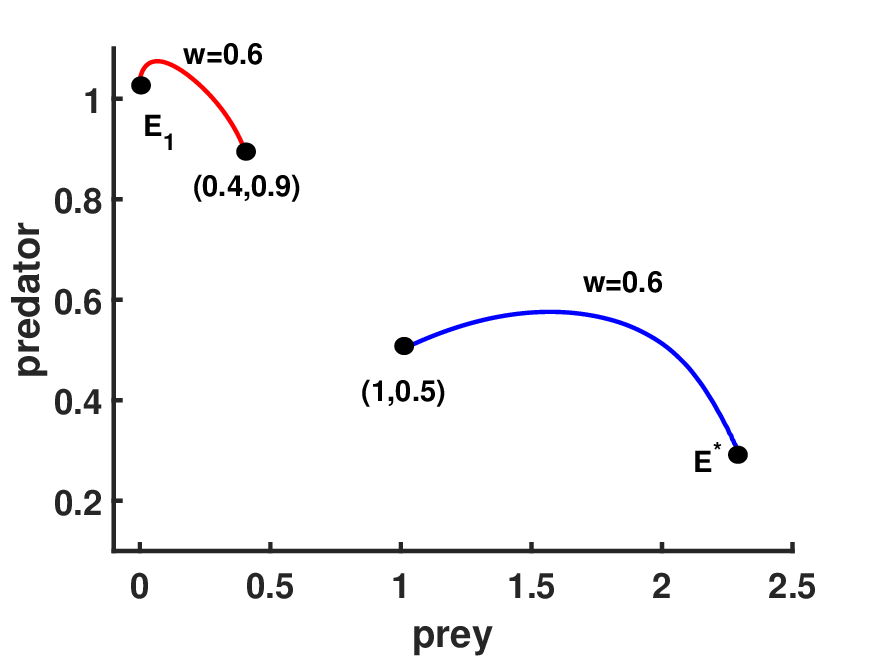}\label{square_6}}
    
\caption{We choose an bounded  function  $\phi(w)$=$\dfrac{2w}{\sqrt{1+w^4}}$. Figs. \ref{square_1} and \ref{square_2}  shows the bifurcation figures of wind speed $w$ vs prey and predator  population, respectively. In Figs. \ref{square_1} and \ref{square_2}, dotted red lines and solid red lines represents unstable and stable equilibrium, respectively and solid blue line represents the function $\phi(w)$. Figs. \ref{square_3} and \ref{square_4} shows the time series for prey and predator population with different $w$ and fig. \ref{square_5} is the phase portrait of the respective population dynamics. The others parameters are r=0.5, k=2.5, b=0.5, d=0.1, $\xi=1.5,$ $\beta=0.285,$ $\alpha=0.5,$ c=0.1. SN: Saddle-Node, TC: Transcritical.}
      
\end{figure}
\begin{figure}[ht!]
	\centering
	\subfloat[]{\includegraphics[width=3 in, height=1.5in]{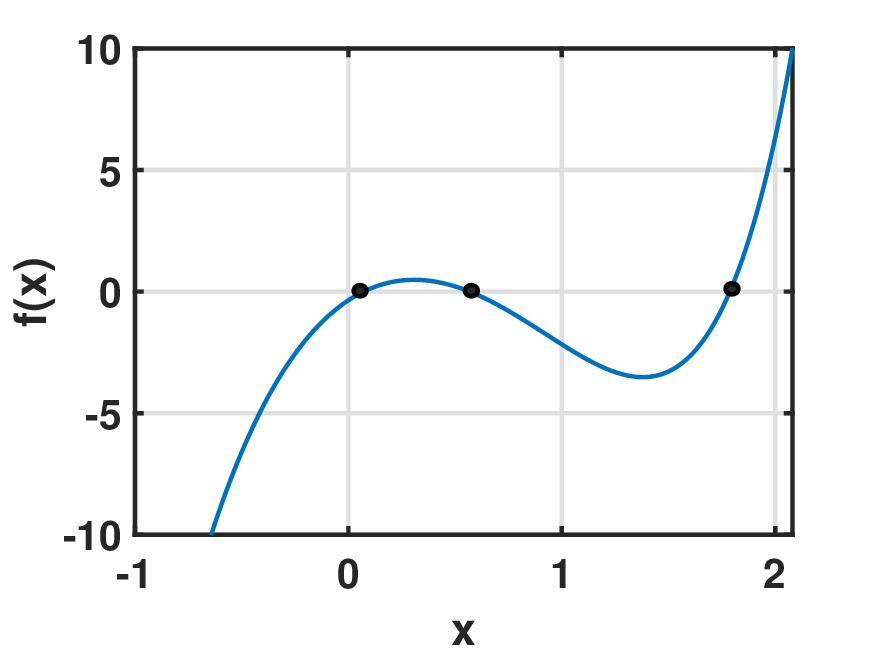}}
    \subfloat[]{\includegraphics[width= 3in, height=1.5 in]{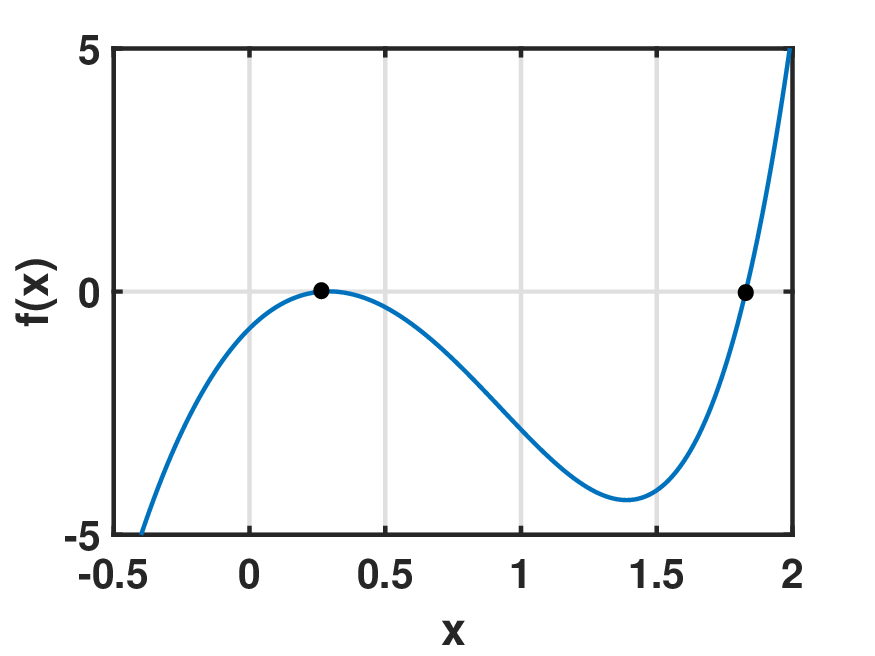}}\\
    \subfloat[]{\includegraphics[width=3 in, height=1.5 in]{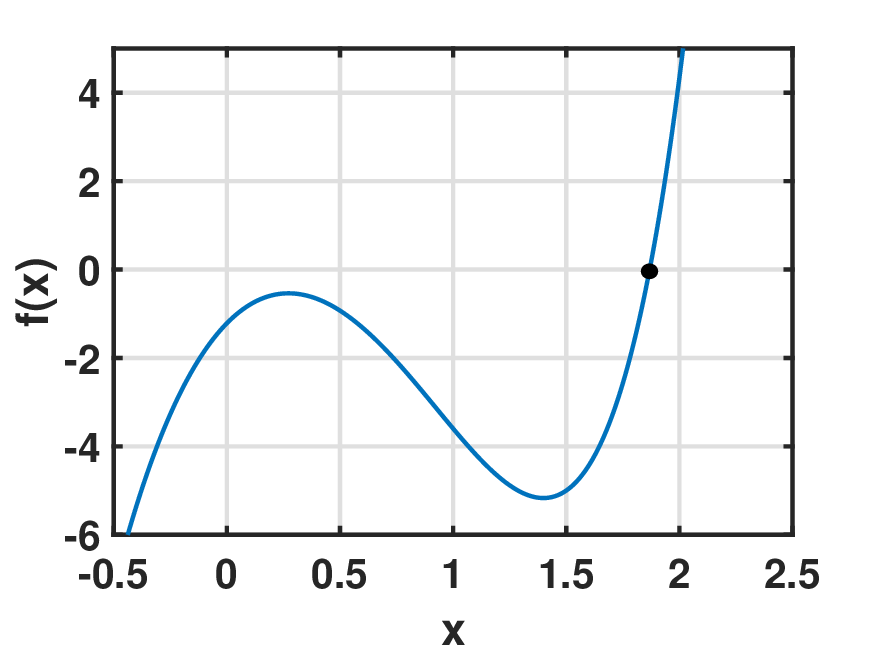}}
    \subfloat[]{\includegraphics[width=3 in, height=1.5 in]{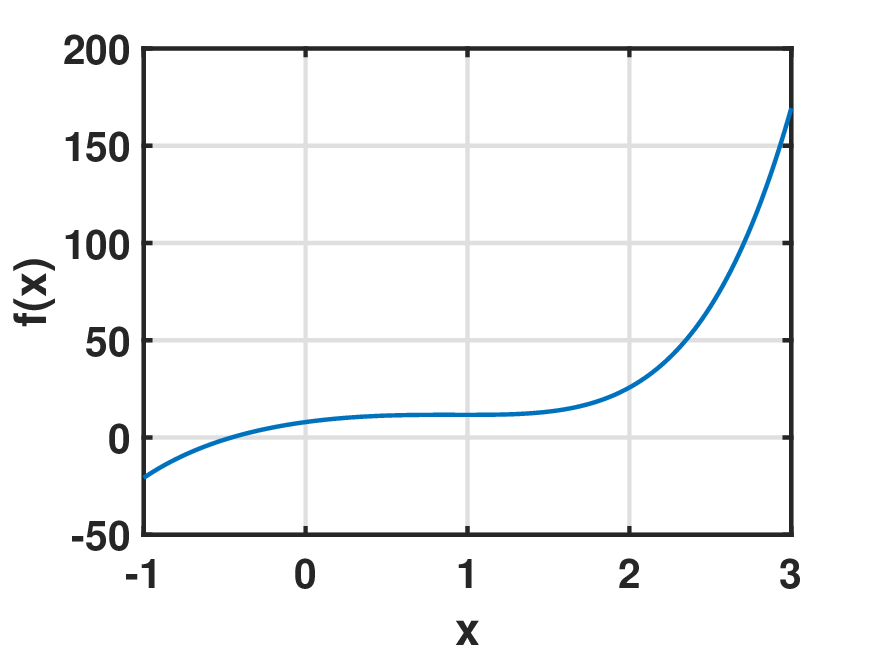}} 
    \caption{ The figure represents the nullcline curve as in equation \eqref{eq:13} for different values of $\phi(w)$. For $\phi(w)=0.92$, $\phi(w)=0.8932$, $\phi(w)=0.86$ and $\phi(w)=1.5$, the equation \eqref{eq:13} that is $f(x)=0$   has  3 solution, 2 solution, unique solution, no solution respectively on $\mathbb{R}^+$. The others parameters are r=0.5, k=2.5, b=0.5, d=0.1, $\xi=1.5,\hspace{2mm}\beta=0.285,\hspace{2mm}\alpha=0.5,\hspace{2mm} c=0.1.$ }
	\label{fig:1}
\end{figure}

\section{Appendix}
\subsection{Positivity and uniform bounded}\label{S3}
With regard to the positivity of the system represented
by equation \eqref{eq:5}, the following outcome has been obtained.\\
\begin{theorem}
All solutions x(t) and y(t) that satisfy 
system \eqref{eq:5} and initial conditions $x(0)>0$, $y(0) > 0$
exhibit uniform boundedness.\\
\end{theorem} 
\begin{proof}
Using the first equation of the system \eqref{eq:5}, it can
be inferred that\\
\begin{equation}\label{eq:6}
\begin{cases}
     \dfrac{dx}{dt}&\leq rx\left(1-\frac{x}{k}\right).
        \end{cases}
   \end{equation}

The above inequality leads to
\begin{equation}\label{eq:7}
 \lim_{t \to \infty} sup \hspace{2mm} x(t) \leq k .\\
\end{equation}
 Hence, for any positive constant  $\epsilon> 0$ sufficiently small,
there exists a $T_1 > 0$ such that for all $t \geq T_1$,\\
 \begin{equation}\label{eq:8}
    k-\epsilon \leq x(t) \leq k+\epsilon.\\  
 \end{equation}
 For t $\geq T_1$, it follows from the second equation of the system
\eqref{eq:5}  that\\
\begin{equation*}
   \begin{split}
        \dfrac{dy}{dt}&= \dfrac{\beta \phi(w) (x+\xi)y}{1+\alpha \xi + bx + x^2} - dy - \phi(w) y^2,\\
        & \leq y\bigg(\dfrac{\beta \phi(w) (k+\epsilon+\xi)}{1+\alpha \xi + b(k-\epsilon) + (k-\epsilon)^2} - d - \phi(w) y\bigg),\\
        &\leq y\bigg(\dfrac{\beta \phi(w) (k+\epsilon+\xi)}{1+\alpha \xi + b(k-\epsilon) + (k-\epsilon)^2} - d \bigg)\bigg[1-\dfrac{\phi(w)y}{\bigg(\dfrac{\beta \phi(w) (k+\epsilon+\xi)}{1+\alpha \xi + b(k-\epsilon) + (k-\epsilon)^2} - d \bigg)}\bigg].\\
        \end{split} 
    \end{equation*}
    The above inequality leads to \\
\begin{equation}\label{eq:9}
 \lim_{t \to \infty} sup \hspace{2mm} y(t) \leq \dfrac{\beta \phi(w) (k+\epsilon+\xi)}{1+\alpha \xi + b(k-\epsilon) + (k-\epsilon)^2} - d.  \\
\end{equation}
Since $\epsilon$ is a small positive constant, setting $\epsilon \to 0$ in \eqref{eq:9}
leads to\\
 \begin{equation}\label{eq:10}
 \lim_{t \to \infty} sup \hspace{2mm} y(t) \leq \dfrac{\beta \phi(w) (k+\xi)}{1+\alpha \xi + bk + k^2} - d . \\
\end{equation}\\
Hence, from equations \eqref{eq:10} and \eqref{eq:7},  x(t) and y(t) of the system \eqref{eq:5}
with the initial conditions x(0) $>$ 0, y(0)$ >$ 0 are uniformly
bounded. 
\end{proof}
\subsection{Existence of the equilibria}\label{Existence}
With regard to the presence of equilibria in the system \eqref{eq:5}, the following result has been observed.\\

   1. The vanishing equilibrium
point $E_0(0, 0)$ and the predator-free equilibrium point
$E_2(k, 0)$ always exist.\\
2. There is a prey-free equilibrium point $E_1(0,y_1)$ \\
     where $y_1=\dfrac{\beta\xi}{c\xi(1+\alpha\xi)}-\dfrac{d}{c\xi\phi(w)} $ , if
 \begin{equation}\label{eq:11}
 d-\dfrac{\beta\xi\phi(w)}{1+\alpha\xi} < 0.
     \end{equation}\\
     3. Other non-negative interior steady states of \eqref{eq:5} can be solved from the following equations\\
     \begin{equation}\label{eq:12}
         \begin{split}
           f_1(x,y)=r\left(1-\frac{x}{k}\right) - \dfrac{\phi(w)  y}{1+\alpha \xi + bx + x^2}=0,\\
           g_1(x,y)=\dfrac{\beta \phi(w) (x+\xi)}{1+\alpha \xi + bx + x^2} - d - c\xi\phi(w) y=0.
         \end{split}
     \end{equation}
     From the first equation of \eqref{eq:12}, we have\\
\begin{equation*}
        y=\dfrac{r}{\phi(w)}(1+\alpha\xi+bx+x^2)\bigg(1-\dfrac{x}{k}\bigg)\hspace{2mm}  \hspace{2mm}\implies \hspace{2mm} y>0 \hspace{2mm}\forall \hspace{2mm} 0<x<k. 
          \end{equation*}
 Solving  equation \eqref{eq:12}, we get a 5-degree equation of $x$ as follows:\\
\begin{equation}\label{eq:13}
    x^5+Ax^4+Bx^3+Cx^2+Dx+E=0,
\end{equation}\\
where $A=2b-k$,\\
$B=2+ b^2+2\alpha\xi-2bk$,\\
$C=2b+2b\alpha\xi-\dfrac{dk}{rc\xi}-2k-b^2k-2\alpha\xi k$,\\
$D=1+\alpha^2\xi^2+2\alpha\xi+\dfrac{k\beta\phi(w)}{rc\xi}-\dfrac{dkb}{rc\xi}-2kb-2bk\alpha\xi$,\\
$E=-\alpha^2\xi^2k-k-2\alpha\xi k+\dfrac{k\beta\phi(w)\xi -dk-dk\alpha\xi}{rc\xi}$.\\
By numerical methods, we show that equation \eqref{eq:13} has at most 3 roots in $\mathbb{R^+}$. Figure \ref{fig:1}.

\begin{theorem} \label{Th_Existence}
The system \eqref{eq:5} has \\
\begin{itemize}
    \item trivial and predator-free equilibrium points, which always exist,
    \item prey-free equilibrium point if $d-\dfrac{\beta\xi \phi(w)}{1+\alpha\xi}<0$,
    \item no interior equilibrium point if $A,B,C,D,E >0$, 
    \item unique equilibrium point if
        \begin{enumerate}
            \item $E<0 $ and $A,B,C,D>0$,
            \item $E,D<0$ and $C,B,A>0$,
            \item $E,D,C<0$ and $A,B>0$,
            \item $E,D,C,B<0$ and $A>0$,
            \item $A,B,C,D,E<0$,
        \end{enumerate}
\end{itemize}
where A, B, C, D and E are the coefficients in \eqref{eq:13}.
\end{theorem}

\subsection{Local stability analysis}     
We introduce the Jacobian Matrix of model \eqref{eq:5}, which is calculated at any arbitrary equilibrium point as follows:\\\\
$$ J = \left(
\begin{array}{cc}
f_x& f_y\\
g_x& g_y\\
\end{array}
\right),$$\\
where $f_x$= $r\left(1-\frac{x}{k}\right)-\frac{rx}{k} - \dfrac{\phi(w)  y}{1+\alpha \xi + bx + x^2}+\dfrac{\phi(w)  xy(b+2x)}{(1+\alpha \xi + bx + x^2)^2}$,\\
$f_y$= $- \dfrac{\phi(w)  x}{1+\alpha \xi + bx + x^2}$,\\
$g_x$= $ \dfrac{\beta\phi(w)  y}{1+\alpha \xi + bx + x^2}- \dfrac{\beta\phi(w)(x+\xi)  y(b+2x)}{(1+\alpha \xi + bx + x^2)^2}$,\\
$g_y$= $ \beta\dfrac{\phi(w) (x+\xi)}{1+\alpha \xi + bx + x^2}-d- 2c\xi\phi(w)y.$\\\\
\begin{lemma}\label{lemma:4.3}
 The trivial Equilibrium Point of the system \eqref{eq:5} is always a saddle point.   
\end{lemma}
\begin{proof}
At equilibrium $E_0(0, 0)$, the Jacobian matrix is\\
$$ J_o = \left(
\begin{array}{cc}
r& 0\\
0& -d\\
\end{array}
\right)$$\\

with eigenvalues:
$\lambda_1 = r$, $ \lambda_2 = -d$.
Since$\lambda_1 >0$ and $\lambda_2 <0$, $E_0(0, 0)$ is a saddle point. 
\end{proof}
\begin{lemma}\label{lemma:4.4}
For the predator-free Equilibrium Point  $E_2(k, 0)$ \\
\begin{itemize}
    \item  It is a stable node if
    \[\phi(w)<\dfrac{d(1+\alpha \xi + bk + k^2)}{\beta(k+\xi)},
    \]
    \item It is a saddle point if
    \[
    \phi(w)>\dfrac{d(1+\alpha \xi + bk + k^2)}{\beta(k+\xi)} ,
    \]
    \item It is non-hyperbolic if
    \[
    \phi(w)=\dfrac{d(1+\alpha \xi + bk + k^2)}{\beta(k+\xi)}. 
    \]
\end{itemize}
\end{lemma}
\begin{proof}
     At equilibrium $E_2(k, 0)$, the Jacobian matrix is\\
$$ J_2 = \left(
\begin{array}{cc}
-r&-\dfrac{\phi(w)k}{1+\alpha \xi + bk + k^2}\\
0&-d+\dfrac{\beta\phi(w)(k+\xi)}{1+\alpha \xi + bk + k^2}
\end{array}
\right)$$\\
with eigenvalues : $\lambda_1 =- r$, $ \lambda_2 = -d+\dfrac{\beta\phi(w)(k+\xi)}{1+\alpha \xi + bk + k^2}$.\\
Since $\lambda_1<0$, the stability depends on $\lambda_2$.\\
$\bullet$ If $\phi(w)<\dfrac{d(1+\alpha \xi + bk + k^2)}{\beta(k+\xi)}$, then $\lambda_2<0$. Hence $E_2(k,0)$ is a stable node.\\
$\bullet$ If $\phi(w)>\dfrac{d(1+\alpha \xi + bk + k^2)}{\beta(k+\xi)}$, then $\lambda_2>0$. Hence $E_2(k,0)$ is a saddle point.\\
$\bullet$ If $\phi(w)=\dfrac{d(1+\alpha \xi + bk + k^2)}{\beta(k+\xi)}$, then $\lambda_2=0$. Hence $E_2(k,0)$ is neutral.   
\end{proof}

\begin{lemma}\label{lemma:4.5}
 For the prey-free Equilibrium Point  $E_1(0, y_1)$ \\
\begin{itemize}
    \item  It is a stable node if
    \[\phi(w)>\dfrac{\bigg(d+rc\xi(1+\alpha\xi)\bigg)(1+\alpha\xi)}{\beta\xi},
    \]
    \item It is a saddle point if
    \[
    \phi(w)<\dfrac{\bigg(d+rc\xi(1+\alpha\xi)\bigg)(1+\alpha\xi)}{\beta\xi} ,
    \]
    \item It is non-hyperbolic if
    \[
    \phi(w)=\dfrac{\bigg(d+rc\xi(1+\alpha\xi)\bigg)(1+\alpha\xi)}{\beta\xi}. 
    \]
\end{itemize}   
\end{lemma}
\begin{proof}
    At equilibrium $E_1(0, y_1)$, the Jacobian matrix is\\
$$ J_1 = \left(
\begin{array}{cc}
r-\dfrac{1}{c\xi(1+\alpha\xi)}\bigg(\dfrac{\beta\phi(w)\xi}{1+\alpha\xi}-d\bigg)& 0\\
\dfrac{\beta}{c\xi(1+\alpha\xi)}\bigg(\dfrac{\beta\phi(w)\xi}{1+\alpha\xi}-d\bigg)-\dfrac{\beta\xi b}{c\xi(1+\alpha\xi)^2}\bigg(\dfrac{\beta\phi(w)\xi}{1+\alpha\xi}-d\bigg)& d-\dfrac{\beta\xi \phi(w)}{1+\alpha\xi}\\
\end{array}
\right).$$\\
Here we are describing the condition of stability of the prey-free equilibrium point $E_1$ of model \eqref{eq:5}.
So the characteristic equation at $E_1(0, y_1)$ is given below,\\
\begin{equation}\label{eq:14}
   \begin{split}
     P(\lambda)=\lambda^2+\mu_1\lambda+\mu_0=0,
    \end{split}
\end{equation}\\
where coefficients are\\
\begin{align*}
    \boldsymbol {\mu_1}&= -tr(J_1)=-(a_{11} + a_{22} ) ,\\
\boldsymbol {\mu_0}&=det(J_1)= a_{11}a_{22} - a_{12}a_{21},\\
\end{align*}\\
where  $ J_1 = \left(
\begin{array}{cc}
a_{11} & a_{12} \\
a_{21} & a_{22} \\
\end{array}
\right)$
is the Jacobian matrix of model \eqref{eq:5} at $E_1$.
The components of the Jacobian matrix are given below:\\
\begin{align*}
    {J_{11}}&=r-\dfrac{1}{c\xi(1+\alpha\xi)}\bigg(\dfrac{\beta\phi(w)\xi}{1+\alpha\xi}-d\bigg),\\
    {J_{12}}&=0,\\
    {J_{21}}&=\dfrac{\beta}{c\xi(1+\alpha\xi)}\bigg(\dfrac{\beta\phi(w)\xi}{1+\alpha\xi}-d\bigg)-\dfrac{\beta\xi b}{c\xi(1+\alpha\xi)^2}\bigg(\dfrac{\beta\phi(w)\xi}{1+\alpha\xi}-d\bigg),\\
    {J_{22}}&= d-\dfrac{\beta\xi \phi(w)}{1+\alpha\xi}.\\
\end{align*}\\
Here $\mu_1^2-4\mu_0=({a_{11}}-{a_{22}})^2\geq0$. Therefore, all the roots of equation \eqref{eq:9} are real.\\
If $tr(J_1)=(a_{11}+a_{22})<0$ i.e. $\mu_1>0$ and $det(J_1)=a_{11}a_{22}>0$ then both roots are real negative. Hence $E_1$ is a stable node.\\
As $a_{22}<0$, $a_{11}a_{22}>0$ implies $a_{11}<0$, that is,\\
\begin{align*}
    \phi(w)>\dfrac{\bigg(d+rc\xi(1+\alpha\xi)\bigg)(1+\alpha\xi)}{\beta\xi}.
\end{align*}\\
If  $det(J_1)=a_{11}a_{22}<0$ then one  positive and one negative real roots. Hence $E_1$ is a saddle point.\\
As $a_{22}<0$, $a_{11}a_{22}<0$ implies  $a_{11}>0$, that is, \\
\begin{align*}
    \phi(w)<\dfrac{\bigg(d+rc\xi(1+\alpha\xi)\bigg)(1+\alpha\xi)}{\beta\xi}.
\end{align*}\\
If \begin{align*}
    \phi(w)=\dfrac{\bigg(d+rc\xi(1+\alpha\xi)\bigg)(1+\alpha\xi)}{\beta\xi},
\end{align*} 
then $E_1$ is neutral.
\end{proof}
\begin{lemma}\label{lemma:4.6}
For the interior equilibrium point $E^*(x^*,y^*)$\\
     \begin{itemize}
    \item  It is  asymptotically stable if $tr(J^*)<0 \hspace{2mm} and \hspace{2mm}  det(J^*)>0,$
    \item It is non-hyperbolic if $
   det(J^*)=0.$
\end{itemize} 
\end{lemma}
\begin{proof}
    At the coexisting equilibrium point $E^*(x^*,y^*)$, the Jacobian matrix is 
$ J^* = \left(
\begin{array}{cc}
J_{11} & J_{12} \\
J_{21} & J_{22} \\
\end{array}
\right)$, \\
where $J_{11}=-\dfrac{rx^*}{k}+\dfrac{\phi(w)x^*y^*(b+2x^*)}{(1+\alpha\xi+bx^*+{x^*}^2)^2}$, \\
$J_{12}=-\dfrac{\phi(w)x^*}{1+\alpha\xi+bx^*+{x^*}^2}$, \\
$J_{21}=\dfrac{\beta\phi(w)y^*}{1+\alpha\xi+bx^*+{x^*}^2}-\dfrac{\beta\phi(w)(x^*+\xi)y^*(b+2x^*)}{(1+\alpha\xi+bx^*+{x^*}^2)^2}$, \\
$J_{22}=-c\xi\phi(w)y^*$.\\
Here, we describe the stability condition of the coexisting equilibrium point $E^*$ of model \eqref{eq:5}.
So, the characteristic equation at $E^*(x^*,y^*)$ is given below,\\
\begin{equation}\label{eq:15}
   \begin{split}
     Q(\lambda)=\lambda^2+\eta_1\lambda+\eta_0=0,
    \end{split}
\end{equation}\\
where coefficients are\\
\begin{align*}
    \boldsymbol {\eta_1}&= -tr(J_1)=-(J_{11} + J_{22} ) ,\\
\boldsymbol {\eta_0}&=det(J_1)= J_{11}J_{22} - J_{12}J_{21}.\\
\end{align*}\\
Here,\\
$tr(J^*)=\dfrac{\phi(w)x^*y^*(b+2x^*)}{(1+\alpha\xi+bx^*+x^{*2})^2}-\dfrac{rx^*}{k}-c\xi\phi(w)y^*$ and\\ 
$det(J^*)=c\xi\phi(w)y^*\bigg[-\dfrac{\phi(w)x^*y^*(b+2x^*)}{(1+\alpha\xi+bx^*+x^{*2})^2}+\dfrac{rx^*}{k}\bigg]+\dfrac{\beta\phi(w)^2x^*y^*}{(1+\alpha\xi+bx^*+x^{*2})^2}\bigg[1-\dfrac{(x^*+\xi)(b+2x^*)}{(1+\alpha\xi+bx^*+x^{*2})}\bigg]$.\\

If $trJ^*<0$ and $detJ^*>0$, then $E^*$ is locally asymptotically stable.\\
If $det(J^*)=0$, then one of eigenvalue of $J^*$ is zero. Hence, $E^*$ is non-hyperbolic. 
\end{proof}
Hence from lemma $\ref{lemma:4.3}$, $\ref{lemma:4.4}$, $\ref{lemma:4.5}$ and $\ref{lemma:4.6}$, a theorem has been formulated as below:\\
\begin{theorem}[Stability of equilibria]Consider the dynamical system \eqref{eq:5} with the parameter set \{$b,c,d,r, k, \alpha, \beta, \xi$, $\phi(w)$\}. The system has the following equilibria:\\
\begin{enumerate}
    \item  The trivial Equilibrium Point $E_0(0,0)$  is always a saddle point.
    \item For the predator-free equilibrium point $E_1(0,y_1)$, stability depends on $\phi(w).$
    \[
    \phi(w)\begin{cases}
        >\dfrac{\bigg(d+rc\xi(1+\alpha\xi)\bigg)(1+\alpha\xi)}{\beta\xi} , & \text{stable},\\[2mm]
        = \dfrac{\bigg(d+rc\xi(1+\alpha\xi)\bigg)(1+\alpha\xi)}{\beta\xi}, & \text{non-hyperbolic},\\[1mm]
        < \dfrac{\bigg(d+rc\xi(1+\alpha\xi)\bigg)(1+\alpha\xi)}{\beta\xi}, & \text{saddle}.
    \end{cases} 
    \]
    \item For the prey-free equilibrium point $E_2(k,0)$, stability depends on $\phi(w).$
    \[
    \phi(w)\begin{cases}
        <\dfrac{ d(1+\alpha \xi + bk + k^2)}{\beta(k+\xi)}, & \text{stable},\\[2mm]
        = \dfrac{ d(1+\alpha \xi + bk + k^2)}{\beta(k+\xi)}, & \text{non-hyperbolic},\\[1mm]
        > \dfrac{ d(1+\alpha \xi + bk + k^2)}{\beta(k+\xi)}, & \text{saddle}.
    \end{cases} 
    \]
    \item For interior equilibrium point $E^*(x^*,y^*)$,
    \[
    \begin{cases}
      tr(J^*)<0 \hspace{2mm}  and \hspace{2mm} det(J^*)>0, & \text{stable},\\[2mm]
      det(J^*)=0,& \text{non-hyperbolic},\\[2mm]
      tr(J^*)>0 \hspace{2mm}  or \hspace{2mm} det(J^*)<0,& \text{unstable}.
       \end{cases}
    \]
\end{enumerate}
\end{theorem}



	\bibliographystyle{plain}
    \bibliography{biblio.bib}
\end{document}